\PassOptionsToPackage{colorlinks,linkcolor={blue},citecolor={blue},urlcolor={red},breaklinks=true,final}{hyperref}
\PassOptionsToPackage{x11names}{xcolor}
\documentclass[acmsmall,screen,nonacm,authorversion]{acmart}

\setcopyright{cc}
\setcctype{by}
\acmDOI{10.1145/3747528}
\acmYear{2025}
\acmJournal{PACMPL}
\acmVolume{9}
\acmNumber{ICFP}
\acmArticle{259}
\acmMonth{8}
\received{2025-02-27}
\received[accepted]{2025-06-27}

\setcitestyle{nosort}
\usepackage{booktabs}
\usepackage{subcaption}
\usepackage{sansmath}

\usepackage{algorithm}
\usepackage{algpseudocode}

\usepackage{stel-common}

\providecommand{\catname}{\mathbf} 
\providecommand{\clsname}{\mathcal}
\providecommand{\oname}[1]{{\operatorname{\mathsf{#1}}}}

\def\defcatname#1{\expandafter\def\csname B#1\endcsname{\catname{#1}}}
\def\defcatnames#1{\ifx#1\defcatnames\else\defcatname#1\expandafter\defcatnames\fi}
\defcatnames ABCDEFGHIJKLMNOPQRSTUVWXYZ\defcatnames

\def\defclsname#1{\expandafter\def\csname C#1\endcsname{\clsname{#1}}}
\def\defclsnames#1{\ifx#1\defclsnames\else\defclsname#1\expandafter\defclsnames\fi}
\defclsnames ABCDEFGHIJKLMNOPQRSTUVWXYZ\defclsnames

\def\defbbname#1{\expandafter\def\csname BB#1\endcsname{{\bm{\mathsf{#1}}}}}
\def\defbbnames#1{\ifx#1\defbbnames\else\defbbname#1\expandafter\defbbnames\fi}
\defbbnames ABCDEFGHIJKLMNOPQRSTUVWXYZ\defbbnames

\DeclareOldFontCommand{\bf}{\normalfont\bfseries}{\mathbf}

\providecommand{\op}{\mathsf{op}}

\providecommand{\dar}{\kern-1.2pt\operatorname{\downarrow}}	
\providecommand{\uar}{\kern-1.2pt\operatorname{\uparrow}}	

\usepackage{stmaryrd}

\providecommand{\lsem}{\llbracket}
\providecommand{\rsem}{\rrbracket}
\providecommand{\sem}[1]{\lsem #1 \rsem}

\providecommand{\pacman}[1]{}					                     

\newcommand{\undefine}[1]{\let #1\relax}					                       

\providecommand{\mone}{{\text{\kern.5pt\rmfamily-}\mathsf{\kern-.5pt1}}}

\makeatletter
\def\mfix#1{\oname{#1}\@ifnextchar\bgroup\@mfix{}}	       
\def\@mfix#1{#1\@ifnextchar\bgroup\mfix{}}			           
\makeatother

\providecommand{\case}[3]{\mfix{case}{\mathbin{}#1}{of}{#2}{\kern-1pt;}{\mathbin{}#3}}

\DeclareMathSymbol{\mathinvertedexclamationmark}{\mathord}{operators}{'074}
\DeclareMathSymbol{\mathexclamationmark}{\mathord}{operators}{'041}
\makeatletter
\newcommand{\raisedmathinvertedexclamationmark}{%
  \mathord{\mathpalette\raised@mathinvertedexclamationmark\relax}%
}
\newcommand{\raised@mathinvertedexclamationmark}[2]{%
  \raisebox{\depth}{$\m@th#1\mathinvertedexclamationmark$}%
}
\makeatother

\newcommand{\Pow}{\mathcal{P}}

\renewcommand{\S}{{\mathcal{S}}}

\newcommand{\takeout}[1]{\empty}

\renewcommand{\rho}{\varrho}

\usepackage{proof}
\renewcommand{\inference}[3][]{\infer[#1]{~#3~}{~#2~}}

\usepackage{accents}

\usepackage[capitalize,noabbrev,nameinlink]{cleveref} 

\usepackage{enumitem}
\setlist[enumerate,1]{label=(\arabic*),font=\normalfont,align=left,leftmargin=0pt,labelindent=0pt,listparindent=\parindent,labelwidth=0pt,itemindent=!,topsep=2pt,parsep=0pt,itemsep=2pt,start=1}
\setlist[enumerate,2]{label=(\alph*),font=\normalfont,labelindent=*,leftmargin=*,start=1}
\setlist[itemize]{labelindent=*,leftmargin=*}
\setlist[description]{labelindent=*,leftmargin=*,itemindent=-1 em}

\usepackage{ifdraft}
\usepackage{seqsplit}
\usepackage{xstring}

\makeatletter
\newcommand{\pushright}[1]{\ifmeasuring@#1\else\omit\hfill$\displaystyle#1$\fi\ignorespaces}
\newcommand{\pushleft}[1]{\ifmeasuring@#1\else\omit$\displaystyle#1$\hfill\fi\ignorespaces}
\makeatother

\tikzstyle{shiftarr}=[
        rounded corners,%
        to path={--([#1]\tikztostart.center)
                     -- ([#1]\tikztotarget.center) \tikztonodes
                     -- (\tikztotarget)},
]

\tikzset{
    commutative diagrams/.cd,
    arrow style=tikz,
    diagrams={>=Straight Barb},
    row sep=large,
    column sep = huge
}

\usetikzlibrary{calc,decorations.pathmorphing,shapes,arrows,bending}

\theoremstyle{definition}

\usepackage{graphicx,scalerel}

\numberwithin{equation}{section}
\newcommand{\Var}{\mathop{\mathsf{Var}}}
\newcommand{\Expr}{\mathop{\mathsf{Expr}}}

\let\xmpsto=\xmapsto
\renewcommand{\xmapsto}[1]{\xmpsto{~#1~}}

\theoremstyle{definition}
\newtheorem{notation}[theorem]{Notation}
\newtheorem{assumptions}[theorem]{Assumptions}

\newcommand{\stepsto}{\longrightarrow}
\newcommand{\xstepsto}[1]{\xrightarrow{~ #1 ~}}
\newcommand{\xstepstostar}[1]{\xstepsto{#1}\mathrel{\vphantom{\to}^{\ast}}}

\renewcommand{\leadsto}{\,\begin{tikzpicture} \draw [
  line width=0.7pt,
  -{Classical TikZ Rightarrow[width=4.4pt,flex=1]},
  line join=round,
  line width=0.6pt,
  decorate, decoration={
    zigzag,
    segment length=5,
    amplitude=1.5,
    pre length=1.5pt,
    post=lineto,
    post length=2.5pt
}] (0,0) -- (0.50,0); \end{tikzpicture}
}

\newcommand{\xleadsto}[1]{%
    \mathrel{%
        \begin{tikzpicture}[%
            baseline={(current bounding box.south)}
            ]
        \node[%
            ,inner sep=.54ex
            ,align=center
            ] (tmp) {~$\scriptstyle #1$~};
        \path[%
            ,draw
            ,line width=0.6pt
            ,-{Classical TikZ Rightarrow[width=4.4pt,flex=1]}
            ,line join=round
            ,decorate,decoration={%
                ,zigzag
                ,segment length=5
                ,amplitude=1.5
                ,pre length=1.5pt
                ,post=lineto
                ,post length=2.5pt
                }
            ] 
        (tmp.south west) -- (tmp.south east);
        \end{tikzpicture}
        }
}

\newcommand{\leadstoalt}{\,\begin{tikzpicture} \draw [
      line width=0.7pt,
      -{Classical TikZ Rightarrow[width=4.4pt,flex=1]},
      line width=1.0pt,
      line join=round,
      decorate, decoration={
        zigzag,
        segment length=5,
        amplitude=1.5,
        pre length=1.5pt,
        post=lineto,
        post length=2.5pt
    }] (0,0) -- (0.50,0); \end{tikzpicture}
}
    
\newcommand{\xleadstoalt}[1]{%
    \mathrel{%
        \begin{tikzpicture}[%
            baseline={(current bounding box.south)}
            ]
        \node[%
            ,inner sep=.54ex
            ,align=center
            ] (tmp) {~$\scriptstyle #1$~};
        \path[%
            ,draw
            ,line width=1.0pt
            ,-{Classical TikZ Rightarrow[width=4.4pt,flex=1]}
            ,line join=round
            ,decorate,decoration={%
                ,zigzag
                ,segment length=5
                ,amplitude=1.5
                ,pre length=1.5pt
                ,post=lineto
                ,post length=2.5pt
                }
            ] 
        (tmp.south west) -- (tmp.south east);
        \end{tikzpicture}
        }
}

\newcommand{\xleadstostar}[1]{\xleadsto{#1}\mathrel{\vphantom{\to}^{\ast}}}

\newcommand{\stopcase}{$\hfill\blacktriangleleft$}
\renewcommand{\S}{\Sigma}
\newcommand{\G}{\Gamma}
\newcommand{\B}{\beta}
\newcommand{\Cfg}{\mathtt{C}}

\renewcommand{\CC}{\mathcal{C}}
\newcommand{\TT}{\mathcal{T}}

\newcommand{\QQ}{\mathcal{Q}}
\newcommand{\RR}{\mathcal{R}}

\newcommand{\OO}{\mathcal{O}}
\newcommand{\FF}{\mathcal{F}}
\newcommand{\GG}{\mathcal{G}}

\newcommand{\setOf}[1]{\{ #1 \}}
\newcommand{\defeq}{\overset{\mathrm{def}}{=}}
\renewcommand{\trace}{\mathsf{Trace}}
\newcommand{\wtrace}{\mathsf{WTrace}}

\newcommand{\subG}{\! {}_{\GG}}

\definecolor{Tawny}{rgb}{0.835, 0.369, 0.0}
\newcommand{\opsem}[1]{\boldsymbol{\textcolor{Tawny}{#1}}}
\newcommand{\opstepsto}{\opsem{\stepsto}}
\newcommand{\xopstepsto}[1]{\opsem{\xstepsto{#1}}}
\newcommand{\opleadsto}{\opsem{\leadstoalt}}
\newcommand{\xopleadsto}[1]{\opsem{\xleadstoalt{#1}}}
\newcommand{\opconfig}[3]{\opsem{\config{#1}{#2}{#3}}}
\newcommand{\opconfigA}{\opconfig{\G}{\S}{\B}}
\newcommand{\opconfigB}{\opconfig{\G'}{\S'}{\B'}}

\newcommand{\opupmap}{\opsem{\upmap}}
\newcommand{\opqmap}{\opsem{\qmap}}
\newcommand{\opemap}{\opsem{\emap}}
\newcommand{\opprmap}{\opsem{\prmap}}

\definecolor{Cerulean}{rgb}{0.0, 0.447, 0.698}
\newcommand{\stsem}[1]{
  \textcolor{Cerulean}{#1}
}
\newcommand{\ststepsto}{
  \stsem{\stepsto}
}
\newcommand{\xststepsto}[1]{
  \stsem{\xstepsto{#1}}
}
\newcommand{\stleadsto}{
  \stsem{\leadsto}
}
\newcommand{\xstleadsto}[1]{
  \stsem{\xleadsto{#1}}
}
\newcommand{\stconfig}[3]{
  \stsem{\config{#1}{#2}{#3}}
}
\newcommand{\stconfigA}{
  \stconfig{\G}{\S}{\B}
}
\newcommand{\stconfigB}{
  \stconfig{\G'}{\S'}{\B'}
}
\newcommand{\stconfigC}{
  \stconfig{\G''}{\S''}{\B''}
}
\newcommand{\stupmap}{
  \stsem{\upmap}
}
\newcommand{\stqmap}{
  \stsem{\qmap}
}
\newcommand{\stmmap}{
  \stsem{\mmap}
}
\newcommand{\stH}{
  \stsem{H}
}

\definecolor{Darkcyan}{rgb}{0.0, 0.62, 0.451}
\newcommand{\quiet}[1]{
  \textcolor{Darkcyan}{#1}
}
\newcommand{\quietleadsto}{
  \quiet{\xleadsto{\tau}}
}
\newcommand{\quietleadstostar}{
  \quiet{\xleadsto{\tau}\mathrel{\vphantom{\to}^{\ast}}}
}
\newcommand{\xquietleadsto}[1]{
  \quiet{\xleadsto{#1}}
}
\newcommand{\xquietstepsto}[1]{
  \quiet{\xstepsto{#1}}
}

\newcommand{\add}{\mathtt{add}}
\newcommand{\rid}{\mathtt{r}}

\newcommand{\dlvr}{\mathtt{dlvr}}

\newcommand{\resp}{\mathtt{ret}}
\newcommand{\qry}{\mathtt{qry}}
\newcommand{\upd}{\mathtt{upd}}

\newcommand{\send}{\mathtt{send}}
\renewcommand{\state}[1]{\langle #1 \rangle}
\newcommand{\opsf}{\mathtt{op}}
\newcommand{\stsf}{\mathsf{st}}

\newcommand{\config}[3]{\langle #1 \mid #2 \mid #3 \rangle}
\newcommand{\replicas}{\mathtt{RID}}
\newcommand{\initstate}{\mathtt{init}}
\newcommand{\darrow}{\downarrow}
\newcommand{\Prog}{\mathsf{Prog}}
\renewcommand{\eval}[1]{\mathcal{E} \sem{#1}}
\newcommand{\Op}{\mathsf{Oper}}
\renewcommand{\op}{\mathtt{op}}
\newcommand{\simulates}{\gtrsim}
\newcommand{\simulatedby}{\lesssim}
\newcommand{\wsimulates}{\gtrapprox}
\newcommand{\wsimulatedby}{\lessapprox}
\newcommand{\host}{\mathsf{host}}
\newcommand{\guest}{\mathsf{guest}}
\newcommand{\msg}{\mathtt{msg}}

\newcommand{\OpUpdate}{\hyperref[rule:OpUpdate]{\textsf{ OpUpdate }}}
\newcommand{\OpQuery}{\hyperref[rule:OpQuery]{\textsf{ OpQuery }}}
\newcommand{\OpDeliver}{\hyperref[rule:OpDeliver]{\textsf{ OpDeliver }}}
\newcommand{\StUpdate}{\hyperref[rule:StUpdate]{\textsf{ StUpdate }}}
\newcommand{\StQuery}{\hyperref[rule:StQuery]{\textsf{ StQuery }}}
\newcommand{\StDeliver}{\hyperref[rule:StDeliver]{\textsf{ StDeliver }}}
\newcommand{\StSend}{\hyperref[rule:StSend]{\textsf{ StSend }}}
\newcommand{\StUpdBC}{\hyperref[rule:StUpdBC]{\textsf{ StUpdBC }}}

\newcommand{\hb}{\prec_{\mathsf{hb}}}

\newcommand{\interp}{\stsem{\mathtt{interp}_{S}}}
\newcommand{\concurrent}[2]{#1\, \|\, #2}

\newcommand{\upmap}{\mathsf{update}}
\newcommand{\prmap}{\mathsf{prep}}
\newcommand{\emap}{\mathsf{effect}}
\newcommand{\qmap}{\mathsf{query}}
\newcommand{\mmap}{\mathsf{merge}}
\newcommand{\bcast}{\mathtt{bcast}}

\newcommand{\rmv}[1]{ \setminus \{ #1 \} }

\theoremstyle{plain}

\theoremstyle{definition}
\newtheorem{assumption}[theorem]{Assumption}

\begin{document}\allowdisplaybreaks

\title{CRDT Emulation, Simulation, and Representation Independence}

\author{Nathan Liittschwager}
\orcid{0009-0005-5602-8509}
\affiliation{%
  \institution{University of California, Santa Cruz}
  \city{Santa Cruz}
  \country{USA}
}
\email{nliittsc@ucsc.edu}

\author{Jonathan Castello}
\orcid{0000-0002-8548-3683}
\affiliation{%
  \institution{University of California, Santa Cruz}
  \city{Santa Cruz}
  \country{USA}
}
\email{jcaste14@ucsc.edu}

\author{Stelios Tsampas}
\orcid{0000-0001-8981-2328}
\affiliation{%
  \institution{University of Southern Denmark}
  \city{Odense}
  \country{Denmark}
}
\email{stelios@imada.sdu.dk}

\author{Lindsey Kuper}
\orcid{0000-0002-1374-7715}
\affiliation{%
  \institution{University of California, Santa Cruz}
  \city{Santa Cruz}
  \country{USA}
}
\email{lkuper@ucsc.edu}

\begin{abstract}
Conflict-free replicated data types (CRDTs) are distributed data structures designed for fault tolerance and high availability. CRDTs have historically been
taxonomized into \emph{state-based} CRDTs, in which replicas apply updates locally and periodically broadcast their state to other replicas over the network, and \emph{operation-based} (or \emph{op-based}) CRDTs, in which every state-updating operation is individually broadcast.  In the literature, state-based and op-based CRDTs are considered equivalent due to the existence of algorithms that let them emulate each other, and verification techniques and results that apply to one kind of CRDT are said to apply to the other thanks to this equivalence. However, what it means for state-based and op-based CRDTs to emulate each other has never been made fully precise.  Emulation is nontrivial since state-based and op-based CRDTs place different requirements on the underlying network with regard to both the causal ordering of message delivery, and the granularity of the messages themselves.

We specify and formalize CRDT emulation in terms of \emph{simulation} by modeling CRDTs and their interactions with the network as transition systems.  We show that emulation can be understood as \emph{weak simulations} between the transition systems of the original and emulating CRDT systems, thus closing a gap in the CRDT literature. We precisely characterize which properties of CRDT systems are preserved by our weak simulations, and therefore which properties can be said to be preserved by emulation algorithms.  Finally, we leverage our emulation results to obtain a general \emph{representation independence} result for CRDTs: intuitively, clients of a CRDT cannot tell whether they are interacting with a state-based or op-based CRDT in particular.

\end{abstract}

\ccsdesc[500]{Theory of computation~Semantics and reasoning}
\ccsdesc[500]{Computing methodologies~Distributed computing methodologies}

\keywords{CRDTs, emulation}

\maketitle

\section{Introduction}\label{sec:introduction}

In distributed data storage systems, \emph{data replication} is a ubiquitous mechanism for guarding 
against machine failures and ensuring that data is physically close to far-flung clients.
Informally, replication copies a data object over $n$ spatially separated sites, or \emph{replicas}, 
with each replica acting as an independent copy of the original object.
With replication comes the challenge of ensuring that replicas remain
consistent with one another
in the face of inevitable network partitions and 
clients who demand ``always-on'' access to data.
The gold standard of consistency for such a replicated system is 
\emph{linearizability}~\cite{herlihy-wing-linearizability}, but it can be impractical to implement --- indeed, systems that are prone to network partitions and that prioritize \emph{high availability} of data must necessarily do so at the
expense of linearizability~\cite{gilbert-lynch-cap, gilbert-lynch-cap-perspectives}.

The quest for an optimal point in the availability/consistency trade-off space
has led to the development of \emph{conflict-free replicated data types}
(CRDTs)~\cite{shapiro-crdts, roh-radts, preguica-crdts}, which are data structures designed
for high availability through replication.
CRDTs sacrifice linearizability in favor of the weaker safety property
\emph{strong convergence}~\cite{shapiro-crdts}, which says that replicas that have received and 
applied the same \emph{set} of updates---in any order---will agree in their observable state.
When coupled with the liveness guarantee of \emph{eventual delivery} of updates to replicas, a system of CRDT replicas achieves 
\emph{strong eventual consistency}~\citep{shapiro-crdts}, which says that eventually the observable state of all replicas will agree.
CRDTs have been an active area of research, with considerable attention paid to the specification and verification (of various properties, but especially strong convergence)
of CRDT designs~\cite{burckhardt-rdts-svo,zeller-state-based-verification,gomes-verifying-sec,
gadduci-crdt-semantics,liu-lh-crdts,nair-state-based-verification,
liang-feng-acc,nieto-aneris-op-based,nieto-aneris-state-based}; recent work tackles
automated verification~\cite{nagar-jagannathan-automated-crdt-verification,deporre-verifx} 
and even synthesis of correct-by-construction CRDTs~\cite{laddad-crdt-synthesis}.

In their pioneering work on CRDTs, \citet{shapiro-crdts} taxonomized CRDTs into \emph{state-based} 
CRDTs, in which replicas apply updates locally and periodically broadcast their 
local state (which may be the result of multiple local updates) to other 
replicas over the network, and \emph{operation-based} (or \emph{op-based}) CRDTs, in which every 
state-updating operation is individually broadcast and applied at each replica.\footnote{More recently, 
\citet{almeida-delta-state-crdts} introduced \emph{delta state} CRDTs, an optimization of traditional 
state-based CRDTs in which only state \emph{changes}, rather than entire states, 
must be disseminated over the network.}
In state-based CRDTs, the replica state space is a
join-semilattice, and a replica receiving an update from a remote replica will apply 
the update locally by taking the least upper bound (join) of its local state and the received update.
As a result, the order of received updates is immaterial.
Op-based CRDTs, on the other hand, rely on stronger ordering guarantees 
(in particular, \emph{causal broadcast}~\cite{birman-reliable,birman-lightweight-cbcast}) 
from the underlying network, while requiring concurrently applied updates
to commute.
Both the state-based and op-based approaches result in strong convergence, the defining characteristic of CRDTs.

Most work on CRDT specification and verification focuses exclusively on either 
state-based~\cite{zeller-state-based-verification,gadduci-crdt-semantics,
nair-state-based-verification,timany-trillium,nieto-aneris-state-based,laddad-crdt-synthesis} 
or op-based~\cite{gomes-verifying-sec,nagar-jagannathan-automated-crdt-verification,liu-lh-crdts,
liang-feng-acc,nieto-aneris-op-based} CRDTs.
The justification for this choice is that state-based and op-based CRDTs can \emph{emulate} each other.
\citet{shapiro-crdts} give general algorithms by which one may
construct a state-based CRDT out of a given op-based CRDT, and vice versa.
However, \citeauthor{shapiro-crdts} stop short of
formally defining a notion of emulation and proving that their construction satisfies it.
In particular, \citeauthor{shapiro-crdts} do not formalize any relationship between the 
\emph{observable behaviors} of the original system and the newly constructed system,
beyond pointing out that if the original system exhibits strong eventual consistency, then so does the new one.
This property is insufficient for correctness, since a trivial CRDT which always
provides the same output in response to all inputs is strongly eventually consistent.
Yet the notion that state-based and op-based CRDTs can emulate each other
is frequently appealed to in the literature.
For instance, \citet{nagar-jagannathan-automated-crdt-verification}, in their work on verification of 
op-based CRDTs, write that ``our technique naturally extends to state-based CRDTs since they can be 
emulated by an op-based model,'' and \citet{laddad-crdt-synthesis}, in their work on synthesis of 
state-based CRDTs, write that state-based CRDTs ``can always be translated to op-based CRDTs if necessary.''
It has not been clear whether the emulation algorithms given by \citeauthor{shapiro-crdts} actually do preserve
the properties discussed in these works.
This makes the notion of emulation in CRDTs ``load-bearing'',
and therefore deserving of being made precise.

In this paper, we seek to close this gap in the CRDT literature and formalize
the notion of CRDT emulation. To do so, we model emulation as formal \emph{simulation}
of transition systems, where each transition system models a state-based or op-based
replicated object, with semantics for client-driven updates,
and message passing of updates between individual replicas.
Unlike \citeauthor{shapiro-crdts}, 
we formally model the network behavior as well as the semantics of individual replicas.
Our framework allows us to specify emulation as a kind of \emph{weak simulation} between the original
system (which we call the \emph{host} CRDT system) and a new system constructed from the original (which we call the \emph{guest} CRDT system).
This model of emulation lets us specify and prove that a guest CRDT exhibits all the observable behaviors of
the host CRDT, and vice versa.

The challenge in reasoning about CRDT emulation is that state-based and op-based CRDTs 
place different requirements on the behavior of the underlying network with regard to 
\emph{causal ordering} of messages.  Op-based CRDTs require causal 
message ordering, but state-based CRDTs do not, and emulating an op-based CRDT with a 
state-based CRDT requires confronting this mismatch in network behavior. \citet{shapiro-crdts}'s original construction works by 
essentially implementing the causal ordering mechanism inside the host CRDT itself.
Our simulation results, however, are about an abstract model of CRDTs operating in a network.
To ground our results, we show that a simple, yet expressive, stateful programming language
can interact with our model of CRDTs by invoking update and query commands, obtaining
observable values from the underlying CRDTs. By leveraging our simulation results,
we then show that from the programmer's point of view, one can interchange the underlying host CRDT
for the constructed guest CRDT with no change in observable behavior.

To summarize, we make the following specific contributions:
\begin{itemize}
      \item Using our formalism, we give a precise study and characterization of
            \emph{emulation} between op-based and state-based CRDTs (\Cref{sec:emulation}).
            Our main result says that a state-based (resp. op-based) host CRDT is
            \emph{weakly simulated} by the constructed op-based (resp. state-based) guest CRDT, and vice versa.
            What makes the simulation interesting
            and challenging is the handling of message delivery.
            In the state-based-to-op-based direction, 
            merging a state can be weakly simulated by a sequence of message deliveries,
            and in the op-based-to-state-based direction,
            delivering a message can be weakly simulated by merging a
            carefully chosen state. It turns out that
            asserting the existence of such a state is nontrivial.
      \item We characterize which properties of CRDT systems are preserved by our weak simulations, and therefore characterize the class of properties that are applicable to a state-based (resp. op-based) CRDT as long as they are applicable to an op-based (resp. state-based) CRDT (\Cref{sec:preserved-properties}).
      \item Finally, we leverage our main emulation result to obtain a 
            general \emph{representation independence} result for 
            CRDTs~(\Cref{sec:representation-independence}).
            Informally, our result says that clients of a
            CRDT cannot tell whether they are interacting with a 
            state-based (resp. op-based) host CRDT or its op-based (resp. state-based) guest CRDT.
\end{itemize}
\noindent Before getting into our contributions, we begin with background and motivating examples in~\Cref{sec:preliminaries}, and we formalize the semantics of CRDTs in~\Cref{sec:coalg}. We discuss related work in \Cref{sec:related} and conclude in \Cref{sec:conclusion}.

\section{Preliminaries and Motivation}\label{sec:preliminaries}

In this section, we give background on CRDTs~(\Cref{subsec:crdt-background}) and on simulation~(\Cref{subsec:simulations-background}). We then give two motivating examples~(\Cref{subsec:motivating-example}) to show why whether op-based and state-based CRDTs can emulate each other is more subtle than it might appear at first glance.

We use the following notational conventions:
      \begin{itemize}
        \item Op-based semantics will be typeset in $\opsem{\textbf{orange}}$,
              e.g., $\opleadsto$, $\opstepsto$, $\opconfigA$;
        \item State-based semantics will be typeset in $\stsem{\text{blue}}$,
              e.g., $\stleadsto$, $\ststepsto$, $\stconfigA$;
        \item Silent transitions (in either op-based or state-based CRDT systems) will be typeset in $\quiet{\text{green}}$: $\xquietleadsto{\tau}$.
      \end{itemize}

\subsection{Background on CRDTs}\label{subsec:crdt-background}

Both op-based and state-based CRDTs enjoy \emph{strong eventual consistency} (SEC for short)~\citep{shapiro-crdts}, which is a combination
of a safety guarantee (\emph{strong convergence}) and a liveness guarantee (\emph{eventual delivery}).
Strong convergence says that if any two replicas receive and apply the same set of updates, then they will have the same observable
state. Eventual delivery says that eventually all replicas receive and apply all updates.
Put together, SEC guarantees that each replica converges to the same 
state.
Op-based and state-based CRDTs accomplish this by placing different constraints on
the network, their allowed operations, and their state space.

\subsubsection{Op-Based CRDTs}\label{subsec:op-based-crdts}

Op-based CRDTs are implemented as a distributed system by having each node (a \emph{replica}) 
implement the same object (seen as a tuple consisting of: 
a replica ID $\rid$, a state $s$, and a set of methods).
Op-based CRDTs achieve strong convergence by requiring each replica $\rid$ 
to sequentially (and independently of other replicas)
apply client-supplied operations as part of a two-phase protocol. The first phase, called \emph{prepare-update},
has no side effects and consists of a replica $\rid$ taking a given operation $\op$ and generating a message $m$ 
that contains the effects of the operation, along with auxiliary metadata. The second phase,
called an \emph{effect-update}, applies that message locally. Immediately after, the replica $\rid$
propagates $m$ to all other replicas $\rid' \neq \rid$ by a \emph{reliable causal broadcast} mechanism.
The separate phases are necessary because in more complicated data structures, e.g., the op-based directed
graph CRDT~\citep{shapiro-crdts}, the effect of an operation can depend on the state.
It is typical to consider these two phases as a single atomic action called an \emph{update},
but logically delineate the two phases by modeling prepare-update and effect-update
as a pair of functions, $\prmap$ and $\emap$, 
respectively~\cite{shapiro-crdts,nieto-aneris-op-based,pure-op-based}.
The update procedure is expressed in \cref{fig:op-based-update}.

\algnewcommand\algorithmicon{\textbf{on}}
\algrenewtext{Function}[1]{\algorithmicon\; #1\, :}
\begin{algorithm}[H]
      \caption{The op-based update as described in \citet{shapiro-crdts}. The prepare phase is used
              so that replicas may attach any needed metadata to messages, e.g., a vector clock~\citep{mattern-vector-time,fidge1988timestamps,schmuck-dissertation}
              and the source replica $\rid$ ID to make messages unique, and ordered by
              potential causality~\cite{lamport-clocks}.}\label{fig:op-based-update}
      \begin{algorithmic}[1]
            \Procedure{Op-based Update}{$\rid : \replicas,  s : S, \op: \mathtt{Op}$}
            \State $m \gets \mathsf{prep}(\rid, s, \op)$\Comment{prepare phase}
            \State $s' \gets \emap(m, s)$ \Comment{effect phase}
            \State $\mathsf{put}(s')$
            \State $\mathsf{broadcast}(\rid, m)$ 
            \EndProcedure
      \end{algorithmic}
\end{algorithm}

Reliable causal broadcast must enforce that messages are reliably delivered (i.e., not lost), and are applied at each replica in an order consistent with the \emph{causal order},
in the specific sense of Lamport's \emph{happens-before} 
partial order~\citep{lamport-clocks}, which we denote $\prec_{\mathrm{hb}}$. Two messages $m$ and $m'$ are
related by $\prec_{\mathrm{hb}}$ if one is a potential cause for the other, 
e.g., $m \prec_{\mathrm{hb}} m'$ if the sending of $m$ is a potential cause
for the sending of $m'$ (we give a more precise description in \cref{sec:opsem}). Conversely, two messages $m, m'$ are
considered \emph{concurrent} (written $\concurrent{m}{m'}$) precisely when
they are unrelated by happens-before, i.e., 
$\neg(m \prec_{\mathrm{hb}} m') \land \neg(m' \prec_{\mathrm{hb}} m)$.
In this case, the effects of $m$ and $m'$ must commute.
\Cref{fig:op-based-object} summarizes the specification of an op-based CRDT replica.

\begin{remark}\label{rem:cbcast}
      If message effects are \emph{always} commutative,
      i.e.,
      \begin{equation}\label{eq:always-commutes}
            \forall m, m' \in M\, :\, (\emap(m') \circ \emap(m))(s) = (\emap(m) \circ \emap(m'))(s)
      \end{equation}
      then the requirement that an op-based CRDT be implemented on top of a reliable causal broadcast mechanism
      may be relaxed to simply reliable broadcast --- all replicas eventually receive
      all broadcasted messages, but with no requirements on their order.
\end{remark}

\begin{figure*}[h]
      \begin{align*}
            & \textbf{Parameters :} \\
            & \quad S : \text{states,}\;  \mathtt{Op} : \text{operations},\; 
            M : \text{messages,}\; Q : \text{queries},\; V : \text{values,}\;
             \\
            & \quad  s^{0} \in S : \text{ initial state } \\
            & \textbf{ Functions :}\\
            & \quad \prmap(\rid, \op, s) : 
                        \text{ replica $\rid$ prepares a message 
                              $m$ that contains the effects of $\op$} \\
            & \quad \emap(m,s) : \text{ applies the effects of $m$ to local state $s$}\\
            & \quad \qmap(q, s) : \text{ $q$ queries the state $s$ returning some value $v$} \\
            & \textbf{Assertions :} \\
            & \quad \text{Messages $m$ are handled by reliable causal broadcast mechanism (if needed)} \\
            & \quad \text{$m$, $m'$ are \emph{concurrent}} \implies 
            (\emap(m') \circ \emap(m))(s) = (\emap(m) \circ \emap(m'))(s)
      \end{align*}
      \caption{Specification of an Op-Based CRDT Object}\label{fig:op-based-object}
\end{figure*}

\subsubsection{State-Based CRDTs}

While op-based CRDTs require that updates are 
applied in an order consistent with the causal order,
state-based CRDTs instead require that (i) the local state space $S$
of each replica $\rid$ forms a \emph{join-semilattice} with join operator $\sqcup$,
and (ii) local updates can only make \emph{inflationary}
updates to the current state of the respective replica, in the sense that
$s \leq \upmap(\rid, \op, s)$ for each operation $\op$.
This means that state-based CRDTs do not need a special preparation phase for message passing.
Instead, it suffices to have replicas
copy their local state $s$ and send it to other replicas, which may invoke
a $\mmap$ method to join their local state with $s$ using the join operator $\sqcup$
that the state space $S$ is equipped with.
Since $\mmap$ is implemented in terms of $\sqcup$, it advances the replica's state toward the least upper bound of all replicas' states. \Cref{fig:st-based-object} summarizes the specification of a state-based CRDT replica.

\begin{figure*}[h]
      \begin{align*}
            & \textbf{Parameters :} \\
            & \quad S : \text{states,}\; \mathtt{Op} : \text{operations,}\; 
            Q : \text{queries},\; V : \text{values,}\; \\
            & \quad \sqcup : \text{a join, i.e., $(S, \sqcup)$ is a join-semilattice} \\
            & \quad  s^{0} \in S : \text{ initial state } \\
            & \textbf{ Functions :}\\
            & \quad \upmap(\rid,\op, s) : \text{ inflationary update } \\
            & \quad \mmap(s, s') : \text{ merges states $s$ and $s'$ together, defined in terms of $\sqcup$} \\
            & \quad \qmap(q, s) : \text{ $q$ queries the state $s$ returning some value $v$} \\
            & \textbf{Assertions :} \\
            & \quad \text{$\sqcup : S \times S \to S$ is associative, commutative, idempotent} \\
            & \quad \mmap(s, s') = s \sqcup s' \\
            & \quad \text{$\upmap$ is inflationary: 
            $\forall s, \op\, :\, \upmap(\rid, \op, s) \geq s$, i.e., $s \sqcup \upmap(\rid, \op, s) = \upmap(\rid, \op, s)$.}
      \end{align*}
      \caption{Specification of a State-Based CRDT Object}\label{fig:st-based-object}
\end{figure*}

\subsubsection{Emulation}

\citet{shapiro-crdts} argue that it is possible for an op-based
CRDT to be \emph{emulated} by a state-based CRDT and vice versa. To that end,
they define two transformations, one that constructs an op-based CRDT given
a state-based CRDT, and one that constructs a state-based CRDT given an op-based CRDT.
The transformations are given as a pair of algorithms, which we denote as
a pair of mappings $\GG^{\stsf \to \opsf}$ and $\mathcal{G}^{\opsf \to \stsf}$ 
for the state-to-op-based transformation and the op-to-state-based transformation respectively.
We defer the precise description of $\GG^{\stsf \to \opsf}$ and $\mathcal{G}^{\opsf \to \stsf}$
to \cref{sec:emulation}, but one can think of $\GG^{\stsf \to \opsf}$ as 
taking an object satisfying the specification in \cref{fig:st-based-object} 
and constructing an object that satisfies the specification in
\cref{fig:op-based-object}. Likewise, $\mathcal{G}^{\opsf \to \stsf}$ takes a given op-based
object and constructs a state-based one.
The original
description of the algorithms can be found in \citet{shapiro-crdts}.

We establish the following terminology with respect to emulation, used throughout the paper.

\begin{definition}[Host and Guest Objects]\label{def:host-guest}
      Let $\OO_{\kappa}$ and $\OO_{\lambda}$ denote two objects, and suppose there is a translation
      $\GG^{\kappa \to \lambda}$ between them which constructs $\OO_{\lambda}$ given
      $\OO_{\kappa}$, in the above sense. Then we say $\OO_{\kappa}$ is the \emph{host object},
      and $\OO_{\lambda}$ is the \emph{guest object}.
\end{definition}
Our use of ``host'' and ``guest'' is meant to evoke the fact that the guest object is implemented on 
top of the host object and can make use of the operations that the host object provides.

While \citet{shapiro-crdts} define their translation algorithms and argue that the
resulting objects satisfy SEC, they do not give a formal definition of
what is meant by ``emulation'' other than that the translated CRDT should satisfy SEC.
Indeed, any host CRDT object may be mapped to a trivial guest CRDT
that returns the identity for every input operation and received message. Such a CRDT is indeed SEC, 
but is clearly pathological, and not in the spirit that 
\citeauthor{shapiro-crdts} intend.
Rather, the idea seems to be that emulation should let the user
swap out a given op-based (resp. state-based) host CRDT for the
corresponding state-based (resp. op-based) guest CRDT.
With this apparent understanding, CRDT specification and verification work tends to focus exclusively on either 
state-based 
or op-based CRDTs, with the justification that 
results that apply to one kind of CRDT should 
straightforwardly ``transfer'' or ``generalize'' to the other.
With this in mind, 
we believe emulation is worth making precise.

\subsection{Simulations}\label{subsec:simulations-background}

The semantic model we use in this paper is a labelled transition system (LTS),
or tuples of the form $(S, \Lambda, \stepsto)$, where $S$ is the state space,
$\Lambda$ is a label set, and $\stepsto$ is a subset of $S \times \Lambda \times S$.
called the \emph{transition relation}. We say a state $s$ transitions to a state $s'$
by a $\alpha \in \Lambda$ if $(s, \alpha, s') \in \stepsto$, and we write $s \xstepsto{\alpha} s'$.
We use the letter $\tau$ (and sometimes $\bot$) to denote transitions that are
\emph{silent} or \emph{unobservable} from the point of view of an observer outside the system,
e.g., a client of a distributed replicated system does not observe the message-passing behavior
that takes place inside the system. We sometimes refer to an LTS purely by its
transition relation (e.g., $\stepsto$), if the state space and label set are clear.

We prefer to define transition relations informally using sets of small-step style rules,
and omit the precise definition of the label set $\Lambda$, instead letting $\Lambda$ be
implicitly defined by the rules. For example, given an LTS $\stepsto$, 
we define its \emph{saturation}, or \emph{weak} transition
relation as a starred arrow, e.g. $\stepsto^{\ast}$, as the least relation closed
under the following rules: 
\[
  \inference{}{s \xstepstostar{\tau} s}
  \qquad
  \inference{s \xstepstostar{\tau} s'' & s'' \xstepsto{\tau} s'}{s \xstepstostar{\tau} s'}
  \qquad
  \inference{s \xstepstostar{\tau} s_{1} & s_{1} \xstepsto{\alpha} s_{2} & s_{2} \xstepstostar{\tau} s'}{s \xstepstostar{\alpha} s'}
\]
Note that $\stepsto^{\ast}$ is the reflexive and transitive closure of $\stepsto$ under silent moves.

Given an LTS $(S, \Lambda, \stepsto)$ with a distinguished initial state $s_0 \in s$,
we say a sequence of labels $\alpha_1, \dots, \alpha_n \in \Lambda$ is a \emph{trace}
or a \emph{behavior} of the LTS if there exists a sequence of transitions
\[s_0 \xstepsto{\alpha_1} s_1 \cdots s_{n-1} \xstepsto{\alpha_n} s_n .\]
We refer to the above alternating sequence of
labels and successor states as an \emph{execution}.

From the perspective of an observer, the internal states in LTS's are immaterial,
and comparison of two LTS's is done by comparing their traces. Reasoning about traces can be done via \emph{simulation},
which is widely used to compare both models and implementations of distributed and 
concurrent systems~\citep{lamport-TLA,lynch1988,milner-ccs,burckhardt-rdts-svo,wilcox2015verdi}.

\begin{definition}[Simulation]
    \label{def:sim}
    Let $X$ and $Y$ be sets, and $\Lambda$ a set of labels.
    Let $\stepsto_{X}\,\subseteq X \times \Lambda \times X$ and 
    $\stepsto_{Y}\,\subseteq Y \times \Lambda \times Y$ be a pair of transition systems.
    We say a relation $R \subseteq X \times Y$ is a \emph{simulation} if for every pair $(x, y) \in R$,
    and $\alpha \in \Lambda$:
      \[ x \xstepsto{\alpha}_{X} x' \implies 
      \exists y' \in Y\, :\, y \xstepsto{\alpha}_{Y} y' \land (x', y') \in R. \]
    In this case, we say $y$ \emph{simulates} $x$, or $x$ is \emph{simulated by} $y$.
\end{definition}
We can express simulations diagrammatically as below.
\[\begin{tikzcd}[ampersand replacement=\&,cramped,sep=small]
	x \& y \\
	{x'}
	\arrow["R", dotted, no head, from=1-1, to=1-2]
	\arrow["\alpha"', from=1-1, to=2-1]
\end{tikzcd}
\quad
\textit{implies}
\quad
\begin{tikzcd}[ampersand replacement=\&,cramped,sep=small]
	x \& y \\
	{x'} \& {\exists y'}
	\arrow["R", dotted, no head, from=1-1, to=1-2]
	\arrow["\alpha"', from=1-1, to=2-1]
	\arrow["\alpha", from=1-2, to=2-2]
	\arrow["R", dotted, no head, from=2-1, to=2-2]
\end{tikzcd}\]

Simulations specifically require that each single step in system $X$ 
be matched by an equivalent single step in system $Y$.
Matching two systems in lock-step this way is usually too strict for realistic models of distributed systems.
Instead, one may want to simulate a single coarse-grained step with a larger number of fine-grained steps,
so we prefer \emph{weak simulation}.

\begin{definition}[Weak Simulation]
  \label{def:weaksim}
    We say a relation $R \subseteq X \times Y$ is a \emph{weak simulation} 
    if for every pair $(x, y) \in R$, and $\alpha \in \Lambda$:
    \[ x \xstepsto{\alpha}_{X} x' \implies \exists y' \in Y\, :\, y \xstepstostar{\alpha}_{Y} y' \land (x', y') \in R. \]
\end{definition}

The union of all (weak) simulations is itself a (weak) simulation, therefore it exists, and we define the following.

\begin{definition}[(Weak) Similarity]\label{def:similarity}
Given a pair of transition systems $\stepsto_{X}$, $\stepsto_{Y}$, and states $x \in X$, $y \in Y$, if there exists a (weak) simulation $\RR$ such that $\RR(x,y)$,
then we say $y$ \emph{(weakly) simulates} $x$ and write $x \simulatedby y$ ($x \wsimulatedby y$).
If also $y \simulatedby x$, then we say $x$ and $y$ are \emph{similar} and write $x \simulatedby \simulates y$
($x \wsimulatedby \wsimulates y$).
\end{definition}

\noindent A stronger form of (weak) simulation is (weak) \emph{bi}simulation.

\begin{definition}[Bisimulation]\label{def:bisimulation}
    If the relation $R$ is a (weak) simulation, and its converse $R^{T}$ is
    a (weak) simulation,
    then $R$ is called a (weak) \emph{bisimulation}.
\end{definition}

Like simulation, the union of all (weak) bisimulations is a (weak) bisimulation,
so it exists, and we denote it by $\sim$ ($\approx$).

\noindent Bisimilarity is in some sense ``symmetric'' in that, if $x \sim y$, and $x \xstepsto{\alpha} x'$,
then we can complete the above simulation diagram, \emph{and} if $y \xstepsto{\alpha} y'$, then we can still complete the simulation
diagram, where we interchange the roles of $x$ and $y$. This symmetry condition is quite strong,
since in general $\sim \neq \simulatedby \simulates$, e.g.,
\[\begin{tikzcd}[ampersand replacement=\&,cramped,column sep=small,row sep=tiny]
	\& {x_{1}'} \\
	x \& {x_{1}} \& {x_{2}} \& {\textit{and}} \& y \& {y_{1}} \& {y_{2}}
	\arrow["a", from=2-1, to=1-2]
	\arrow["a", from=2-1, to=2-2]
	\arrow["b", from=2-2, to=2-3]
	\arrow["a", from=2-5, to=2-6]
	\arrow["b", from=2-6, to=2-7]
\end{tikzcd}\]
has $x \simulatedby \simulates y$ but $x \not \sim y$. Indeed,
any sequence of actions in one system is simulated by a sequence of actions in the other system,
but if $x \xstepsto{a} x'_{1}$, then we match by $y \xstepsto{a} y_{1}$, but $x'_{1}$ has no answer
if $y_{1} \xstepsto{b} y_{2}$. This is what we mean by lack of symmetry in the simulation,
and this situation turns out to be the case when modeling op-based and state-based CRDTs.

\subsection{Motivating Examples}\label{subsec:motivating-example}

For an op-based CRDT and a state-based CRDT to be considered ``equivalent'' despite different implementations, we would expect them to be
indistinguishable in terms of their \emph{behavior}. 
A CRDT can be modeled with an LTS, where labels $\alpha \in \Lambda$ correspond to a client-facing 
API (e.g., $\upmap$ and $\qmap$ commands),
so emulation should mean that given an LTS $\stepsto$
modeling a system of op-based (resp. state-based) CRDT replicas,
we can use \citeauthor{shapiro-crdts}'s emulation algorithms to construct
a corresponding LTS $\stepsto'$ (with the same label set $\Lambda$) 
modeling a system of state-based (resp. op-based) CRDT replicas,
such that there is a (weak) bisimulation between their initial states.
However, it turns out that some care is needed, as this approach can fail in subtle ways, which we now demonstrate with two examples using an op-based \emph{grow-only set}
CRDT and its translation into a state-based CRDT using a \citeauthor{shapiro-crdts}-style emulation.

\paragraph{Op-Based Grow-Only Set}
We specify the op-based grow-only set as an \emph{object} in the sense of \cref{fig:op-based-object}. 
The replica states are sets of positive integers, with initial state $\opsem{s}^{0} = \varnothing$,
and the only operation is to add some $n \in \mathbb{N}$ to the
current state (a set). Queries simply sum the elements in the set, and return the result as a value. 
Messages are the operation written as a function: $m(k) \defeq \lambda(s : \text{state}) \, . \,  s \cup \{k\}$,
wrapped in a $\langle - \rangle$ constructor.
\cref{fig:op-based-gset}
defines the $\prmap$, $\emap$, and $\qmap$ functions.

\begin{figure*}[h]
\begin{gather*}
      \Pow(\mathbb{N}) : \text{states,} \qquad 
      \{\add[n] \mid n \in \mathbb{N}\}: \text{operations}, \qquad
      \mathbb{N} : \text{values}, \qquad
      \quad
      \{\mathsf{sum}\} : \text{queries} \\
      \{ \langle \lambda(s' : \text{state}) \, . \,  s \cup \{n\} \rangle \mid n \in \mathbb{N}, \rid \in \replicas\} : \text{messages}, \\
      \opsem{\prmap}(\rid, \add[k], \opsem{s}) = \langle m(k) \rangle \qquad
      \opsem{\emap}(\langle m \rangle,\opsem{s}) = m(\opsem{s}) \qquad
      \opsem{\qmap}(\opsem{s}) = \sum_{n \in \opsem{s}} n
\end{gather*}
\caption{Op-based grow-only set. Updates are defined as a composition of $\prmap$ and $\emap$,
  i.e., $\upmap(\rid, \op, s) = \emap(\prmap(\rid, \op, s), s)$.}\label{fig:op-based-gset}
\end{figure*}

\paragraph{Translation to State-Based Grow-Only Set}
To obtain the corresponding op-based CRDT, we use an emulation algorithm inspired by \citet{shapiro-crdts},
We defer the details to \cref{sec:op-to-state-details}, but the basic idea is
that we use \emph{sets} of messages to represent the internal state, equipped with an \emph{interpretation}
function $\interp$ that maps a set of op-based messages 
$\stsem{H} = \{\langle m_1 \rangle, \dots, \langle m_k \rangle\}$ to an op-based state $\opsem{s}$.
The state-space $(\Pow(M), \cup)$ is a join-semilattice, where $\cup$ is set-theoretic union, and thus acts as
our $\mmap$ function.\footnote{
      This is slightly different from how \citet{shapiro-crdts}
      define it. They maintain a tuple $(s, K, D)$ where $s$
      is the original op-based state, $K$ is the set of ``known'' messages, and $D$ is a set of ``delivered'' messages.
      Then $\sqcup$ is defined by $(s, K, D) \sqcup (s', K', D') = d(s, K \cup K', D)$,
      where $d$ recursively applies all yet unapplied messages in $(K \cup K') \setminus D$, updating $s$ and $D$ accordingly.
      Curiously, this is not actually a join-semilattice,
      since now $\sqcup$ is not commutative! Our modification, on the other hand, is.
}
\cref{fig:st-based-gset} lays out the definitions. The operations, values, and queries are the same, 
so we omit them. Note that all operations commute, so the application order of messages used by $\interp$ is immaterial, and if $\stsem{H} = \varnothing$,
then we just return $\opsem{s}^{0}$. To make it clear when $\stsem{H}$ is acting as a message,
we use a $\stsem{\msg}$ constructor.

\begin{figure*}[h]
\begin{gather*}
      \Pow(\Pow(\mathbb{N}) \to \Pow(\mathbb{N})) : \text{states}\\
      \stsem{\upmap}(\rid, \add[n], \stsem{H}) = \stsem{H} \cup \{\opsem{\prmap}(\rid, \add[n], \varnothing)\} \\
      \stsem{\mmap}(\stsem{H}, \stsem{\msg}\langle \stsem{H'} \rangle) = \stsem{H} \cup \stsem{H'} \qquad
      \stsem{\qmap}(\stsem{H}) = \opsem{\qmap}(\interp(\stsem{H})) \\
      \interp(\{\langle m_1 \rangle, \dots, \langle m_k \rangle\}) = (m_k \circ \cdots \circ m_1)(\opsem{s}^{0})
\end{gather*}
\caption{State-based grow-only set translated from the op-based grow-only set of \Cref{fig:op-based-gset}.}
  \label{fig:st-based-gset}
\end{figure*}

\paragraph{Transition Semantics}
To set up our simulation, we want to embed the objects of \cref{fig:op-based-gset} and \cref{fig:st-based-gset} into labeled transition
systems, using the same label set $\Lambda$ for both to make use of \cref{def:sim}. We
think of $\Lambda$ as a union of an external interface where we restrict the client to either update requests, or queries,
and some additional ``hidden'' labels to model internal communication unobservable to the client. We use the label set
      \[
            \Lambda \defeq \{(\upd[\add , n], \bot) \mid n \in \mathbb{N}\} \cup 
                              \{(\qry, \resp[v]) \mid v \in \mathbb{N}\} \cup \{\quiet{\dlvr}, \quiet{\bcast}\}. \\
      \]
The $(\upd[\add , n], \bot)$ label denotes updates (with $\bot$ indicating no return value), and the $(\qry, \resp[v])$ label denotes queries (with return value $v$); these make up the \emph{external interface}.
The $\quiet{\dlvr}$ and $\quiet{\bcast}$ labels denote message-passing transitions and make up the \emph{internal interface}. We write the state transitions
at the replica level, e.g., $\rid : s \xstepsto{\alpha} s'$ says that replica $\rid$ transitions from $s$ to $s'$, by action $\alpha \in \Lambda$.
We assume in each example a \emph{set} of replicas, and a \emph{message buffer} 
$\beta$ that contains messages generated by the $\quiet{\bcast}$
transitions. The $\quiet{\dlvr}$ transitions remove a message from the buffer. We omit explicit transitions on $\beta$,
preferring to leave them implicit for brevity.
In \cref{sec:coalg}, we develop these transition semantics in full detail. We now give the two examples.


\begin{example}[Message Granularity Breaks Bisimulation]\label{ex:non-bisim}
In this first example, we assume that op-based CRDTs broadcast the effects of their operations \emph{as part of the local update},
and state-based CRDTs broadcast their states \emph{independent of the updates}, i.e., as a separate step. This is a standard
interpretation of the semantics of op-based and state-based CRDTs~\citet{shapiro-crdts}.
With this semantics, a weak bisimulation is \emph{not} possible in general, even though
a \emph{pair} of (one-way) weak simulations is (as we show in \cref{sec:emulation}).

Suppose we have two replicas $\replicas = \{\rid_1, \rid_2\}$.
We denote the op-based replicas by $\opsem{\rid_1, \rid_2}$, and the state-based replicas by $\stsem{\rid_1, \rid_2}$.
Suppose a client has $\stsem{\rid_1}$ perform two updates: $\add[5]$ and $\add[42]$, each invoking $\stsem{\upmap}$ to update its state.
Since $\quiet{\bcast}$ is considered a \emph{separate action}, we could have $\stsem{\rid_1}$ broadcast its state \emph{after} the two updates. So we have
\begin{gather}
      \stsem{\rid_1} : \varnothing \xststepsto{(\upd[\add, 5], \bot)}
      \{ \langle m(5) \rangle \}  \xststepsto{(\upd[\add, 42], \bot)} 
      \{\langle m(5) \rangle, \langle m(42)\rangle \} \xquietstepsto{ \bcast } 
      \{ \langle m(5) \rangle , \langle m(42) \rangle\}, \label{steps1}
\end{gather}
where the message buffer $\beta(\stsem{\rid_2})$ contains only the transmitted state $\stsem{\msg}\,\{\langle m ( 5 ) \rangle, \langle m ( 42 ) \rangle\}$ afterwards.

The op-based replica $\opsem{\rid_2}$ can simulate this
by performing the same operations, but updates and broadcasts happen together as a single atomic step. We express
this as the composed transition $\xopstepsto{\upd} \circ \xquietstepsto{\bcast}$, and we have
\begin{gather}
      \opsem{\rid_1} : \varnothing \xopstepsto{(\upd[\add,5], \bot)} \circ \xquietstepsto{\bcast} 
      \{5\} \xopstepsto{(\upd[\add, 42], \bot)} \circ \xquietstepsto{\bcast} \{5, 42\}, \label{steps2}
\end{gather}
where the message buffer $\beta(\opsem{\rid_2})$ contains the \emph{two} messages $\langle m ( 5 ) \rangle$ and $\langle m ( 42 ) \rangle$.

This seems to work. If $\stsem{\rid_2}$ performs a $\quiet{\dlvr}$ action, removing $\stsem{\msg}\, \{\langle m ( 5 ) \rangle, \langle m ( 42 ) \rangle\}$ from the buffer,
and invoking $\stsem{\mmap}$, then $\opsem{\rid_2}$ can easily simulate this by delivering $\langle m ( 5 ) \rangle$, then $\langle m ( 42 ) \rangle$, invoking $\opsem{\emap}$ for each.
This gives us
\[ \stsem{\rid_2} : \varnothing \xquietstepsto{\dlvr} \{ \langle m ( 5 ) \rangle, \langle m ( 42 ) \rangle\} 
      \quad \implies \quad 
      \opsem{\rid_2} : \varnothing \xquietstepsto{\dlvr} \{5\} \xquietstepsto{\dlvr} \{5, 42\},\]
which is a simulation of $\stsem{\rid_2}$ by $\opsem{\rid_2}$ since a $\qry$ action on either gives the same
return value:
\begin{align*}
      \opsem{\rid_2} :  \{ 5, 42 \} \xopstepsto{(\qry, \resp[47])} \{ 5, 42 \}
      \quad \textit{and} \quad
      \stsem{\rid_2} : \{\langle m ( 5 ) \rangle, \langle m ( 42 ) \rangle\} \xststepsto{(\qry, \resp[47])} \{\langle m ( 5 ) \rangle, \langle m ( 42 ) \rangle\},
\end{align*}
since we have 
\[ \opsem{\qmap}(\{5, 42\}) = 5 + 42 = \opsem{\qmap}((m(42) \circ m(5))(\varnothing)) = \stsem{\qmap}(\{\langle m ( 5 ) \rangle, \langle m ( 42 ) \rangle\}). \]
So what is the problem? The problem is that bisimilar states should be indistinguishable \emph{in either direction}.
The transitions \eqref{steps1} and \eqref{steps2} do not lead to equivalent states, since $\opsem{\rid_2}$ can perform an
action that $\stsem{\rid_2}$ cannot. The issue is the \emph{granularity} of the messages. The op-based buffer contains \emph{two}
messages: $\langle m ( 5 ) \rangle$ and $\langle m ( 42 ) \rangle$, but the state-based buffer contains effectively \emph{one} message --- the state $\{\langle m ( 5 ) \rangle, \langle m ( 42 ) \rangle\}$. This means
$\opsem{\rid_2}$ could do the transition
\[ \opsem{\rid_2} : \varnothing\, \xquietstepsto{\dlvr}\; \opsem{\emap}(\langle m ( 5 ) \rangle, \varnothing) \quad = \quad \{5\},\]
but $\stsem{\rid_2}$ only has available the transition
\[ \stsem{\rid_2} : \varnothing\, \xquietstepsto{\dlvr}\; \stsem{\mmap}(\varnothing, \stsem{\msg}\, \{\langle m ( 5 ) \rangle, \langle m ( 42 ) \rangle\}) \quad = \quad \{\langle m ( 5 ) \rangle, \langle m ( 42 ) \rangle\},\]
and a $\qry$ action would give different results. Despite the fact that there is no bisimulation, note that strong convergence
is not in jeopardy --- if all messages are eventually delivered, then $\opsem{\rid_2}$ and $\stsem{\rid_2}$ do obtain the same observable state.
\end{example}


\begin{example}[Causal Order Cannot Be Ignored]\label{ex:causal-order-required}
      In their original presentation in \citet{shapiro-crdts}, op-based CRDTs are assumed to be built on top of
      reliable causal broadcast, which ensures that updates are propagated and delivered to each replica
      in an order consistent with the causal order. However, op-based CRDTs where \emph{all} updates commute do not require
      causal delivery order for correctness. In those cases, causal delivery would be overkill, and it is typical to relax the assumption of reliable causal broadcast
      to just reliable broadcast. The grow-only set we specified is one such example where all updates
      commute. It turns out that in this case, not even a one-way simulation is possible in general
      (as opposed to \emph{bi}simulation in \cref{ex:non-bisim}), even though both the op-based and state-based implementations can
      \emph{converge} to the same state.

      We now consider three replicas for the grow-only set, and allow messages to be delivered in any order.
      Suppose $\opsem{\rid_1}$ executes the transition 
            $\opsem{\rid_1} : \varnothing\; \xopstepsto{(\upd[\add, 1]), \bot} \circ \xquietstepsto{\bcast}\; \{1\}$.
      The buffers $\beta(\opsem{\rid_2})$ and $\beta(\opsem{\rid_3})$ both now contain the message $\langle m ( 1 ) \rangle$, 
      allowing $\opsem{\rid_2}$ to execute the transitions
      \begin{gather}
            \opsem{\rid_2} : \varnothing\; \xquietstepsto{\dlvr}\; \{1\}\;
            \xopstepsto{ (\upd[\add, 2], \bot)} \circ \xquietstepsto{\bcast}\; \{1, 2\}.
            \label{r2steps}
      \end{gather}
      Then $\stsem{\rid_1}$ and $\stsem{\rid_2}$ can try to simulate this with the transitions
        \begin{align}
            \stsem{\rid_1} & : \varnothing\, 
                \xststepsto{(\upd[\add, 1], \bot)} \{ \langle m ( 1 ) \rangle \} \xquietstepsto{\bcast} \{ \langle m ( 1 ) \rangle \} \\
            \stsem{\rid_2} & : \varnothing\;
                  \xquietstepsto{\dlvr}\; \{ \langle m ( 1 ) \rangle\}
                  \xststepsto{(\upd[\add, 2], \bot)}\; \{\langle m ( 1 ) \rangle, \langle m ( 2 ) \rangle\}\;
                  \xquietstepsto{\bcast}\; \{\langle m ( 1 ) \rangle, \langle m ( 2 ) \rangle\}.
            \label{attemp-sim}
        \end{align}
    
      It is here that we run into the issue. Notice that the sending of $\langle m ( 1 ) \rangle$ 
      causally precedes the sending of $\langle m ( 2 ) \rangle$
      since $\opsem{\rid_2}$ sequentially executes a $\quiet{\dlvr}$ action before the $\upd[\add, 2]$ action
      in \eqref{r2steps}.
      Now, replica $\opsem{\rid_3}$ has both the $\langle m ( 1 ) \rangle$ and $\langle m ( 2 ) \rangle$ messages
      in its input buffer $\beta(\opsem{\rid_3})$.
      Without enforcing causal delivery,
      $\opsem{\rid_3}$ is free to deliver $\langle m ( 2 ) \rangle$ before $\langle m ( 1 ) \rangle$, since it will converge to the same state as $\opsem{\rid_2}$ either way.

      However, $\stsem{\rid_3}$ cannot perfectly simulate what $\opsem{\rid_3}$ does.
      After \eqref{attemp-sim} executes, $\stsem{\rid_3}$'s buffer $\beta(\stsem{\rid_3})$ contains the two states 
      $\stsem{\msg}\, \{\langle m ( 1 ) \rangle\}$, and $\stsem{\msg}\, \{\langle m ( 1 ) \rangle, \langle m ( 2 ) \rangle\}$,
      and delivering either choice (or both) does not simulate $\stsem{\rid_3}$'s intermediate state of $\{2\}$,
      since $\opsem{\qmap}(\{ 2\}) = 2$, but
      \begin{gather*}
            \stsem{\qmap}(\{\langle m ( 1 ) \rangle\}) = 1
            \quad \textit{and} \quad
            \stsem{\qmap}(\{\langle m ( 1 ) \rangle, \langle m ( 2 ) \rangle\}) = 3.
      \end{gather*}

On the other hand, if we required causal delivery, then since the sending of $\langle m ( 1 ) \rangle$ causally precedes the sending of $\langle m ( 2 ) \rangle$,
replica $\opsem{\rid_3}$ would have to deliver $\langle m ( 1 ) \rangle$ before $\langle m ( 2 ) \rangle$. This rules out the problematic
execution, since now $\opsem{\rid_3}$ can only exhibit the behaviors $1$ or $3$, the same as $\stsem{\rid_3}$.
\end{example}


To summarize, what \cref{ex:non-bisim} shows is that a given host CRDT (the op-based grow-only set) and its guest (the constructed state-based grow-only set)
can both \emph{converge} to the same observable state, but are not \emph{indistinguishable}.
The simulation does not preserve the \emph{branching behavior}, since $\opsem{\rid_2}$
can exhibit more possible behaviors than $\stsem{\rid_2}$, even though they are in `equivalent' states.
A weak bisimulation, on the other hand, \emph{does} respect `branching behavior'\citep{Sangiorgi_2011,milner-ccs},
so \cref{ex:non-bisim} would be illegal under a bisimulation.
This difference arises from how we treat message passing --- state-based CRDTs decouple updates and broadcasts,
while op-based CRDTs execute both together (c.f. \cref{fig:op-based-update}).
In some sense, state-based CRDT replicas can accumulate a large piece of state that causes other replicas
to take a `large' transition step when they perform $\mmap$. 

\cref{ex:causal-order-required}, on the other hand, shows that
the causal delivery order assumption cannot be ignored \emph{even for CRDTs for which all operations commute}.
Abandoning the causal ordering requirement means that an op-based CRDT can express intermediate behaviors
that a state-based CRDT cannot simulate at all (in the sense of \cref{def:sim}).
Once again, though, strong convergence holds, since both CRDT systems can converge to the same state. We return to this point of reliable causal broadcast in later sections.

These two examples show how a difference in behavior can arise between CRDTs that `should' be considered equivalent. In the rest of the paper, we tackle this problem, 
laying out more precisely what classes of CRDTs are equivalent under the emulation algorithms.

\section{Formal System Semantics of CRDTs}
\label{sec:coalg}

We now properly formalize our operational semantics for CRDTs.
The semantics are similar to those of \cref{ex:non-bisim,ex:causal-order-required},
with some minor changes for technical reasons.
We start by describing the semantics in a CRDT-agnostic way,
starting with the low-level building blocks of replicas and events,
then higher-level components such as the global system state and configurations.
For the purposes of this section, we fix
a set $\replicas$ of replica identifiers.

\subsection{Replicas and Events}

We think of each replica in a CRDT system as independently implementing a \emph{state machine}
according to the object specification \cref{fig:op-based-object} in the op-based case,
and \cref{fig:st-based-object} in the state-based case.
The state machine is described by the \emph{state space} $S$ with an initial state $s^{0}$,
a set $I$ of \emph{inputs}, and a set $O$ outputs.
We think of $I$ as the union of client-supplied commands, and
the set of possible messages delivered from other replicas.
Outputs $O$ can be either \emph{values} to be returned to the client,
or generated messages intended for delivery at other replicas.
A replica is then simply a state machine $\delta : \replicas \times I \times S \to \{\bot\} + S \times O$
equipped with a current state $s \in S$.
Each replica implements the same $\delta$, and all begin with the initial state $s^{0}$.
When a replica $\rid$ with state $s$ transitions consumes input $i$ and transitions to state $s$, yielding output $o$, we record
the tuple $(\rid, i, o)$ as an \emph{event}.
\begin{definition}[Replica]\label{def:replica}
  A \emph{replica} (of a state machine $\delta$)
  is a fixed identifier $\rid \in \replicas$ equipped with a current state $s \in S$,
  and a labeled transition relation $\stepsto$ defined by the following rule,
  for all $s,s' \in S$, $i \in I$, and $o \in O$.
  \[ \inference[\mathsf{Step}]{
        \delta(\rid, i, s) \neq \bot &
        \delta(\rid, i, s) = (s', o)
      }{
        s \xstepsto{(\rid, i, o)} s'
      }\label{rule:Event} \]
\end{definition}

We simply refer to the identifier set $\replicas$ as a set of \emph{replicas} if each
$\rid \in \replicas$ is a replica of the same state machine. 
We note that replicas are deterministic in the sense that for a fixed $\rid$, and $i$, and $s$, we have
\[ \forall o', o'' \in O, \forall s', s'' \in S\, : \, (s \xstepsto{(\rid, i, o')} s' \land
   s \xstepsto{(\rid, i, o'')} s'') \implies s' = s'' \land o' = o'',\]
and so for that reason, we specify replicas entirely in terms of $\stepsto$ with no loss of generality.

Our events are parameterized by input/output sets $I$ and $O$,
but follow the same syntactic structure for replicas of any
state machine.
\begin{definition}[Inputs/Outputs]\label{def:events}
  Let $\mathsf{Op}$ be a set of operations, $Q$ be a set of queries,
  $V$ a set of values, and $M$ a set of messages, and $\bot$ a ``null'' symbol.
  We define the set of inputs $I$ and outputs
  $O$ according to the following grammar:
  \begin{equation}
    \label{eq:eventsyn}
    \begin{aligned}
      {\quad (\text{Inputs}\ I)\qquad}
      &  i \Coloneqq \bot \mid \upd[\op] \mid \qry[q] \mid \dlvr[m]
      \quad &&  \text{where } \op \in \mathsf{Op}, q \in Q, m \in M.\\
      {\quad (\text{Outputs}\ O)\qquad}
      &  o \Coloneqq \bot \mid \resp[v] \mid \send[m]
      \quad && \text{where } v \in V, m \in M.
    \end{aligned}
  \end{equation}
\end{definition}
We overload $\bot$ as an input to denote a computational step, as well as a lack of output.
The inputs $\upd[\op]$ and $\qry[q]$ denote client-given commands and thus make up the client-facing interface.
The $\dlvr[m]$ inputs denote the consumption of a message $m$ from an external buffer.

A response output $\resp[v]$ denotes a replica producing a value $v$ in response to some query.
The output $\send[m]$ denotes that the message $m$ should be pushed into the network by an external handler
  (e.g., $\mathtt{broadcast}$ as in \cref{fig:op-based-update}).

\subsection{System-Level Preliminaries}\label{subsec:system-preliminaries}

Here we describe the system-level semantics that are common to all CRDT implementations.
Just as individual replicas step following a transition system $\stepsto$,
  a system of replicas will step following a transition system $\leadsto$.
We prefer an interleaving model of concurrency, where a single $(\rid, i, o)$ event labels a single replica-level transition,
  though the choice of which replica steps next is non-deterministic.

We denote the global state of a system (called a \emph{configuration}) with $\config{\Gamma}{\Sigma}{\beta}$.
It consists of three components: the \emph{replica states} denoted by the function
$\Sigma : \replicas \to S$ that sends a replica $\rid$ to its state $s$;
a \emph{network state} or \emph{message buffer} $\beta$ that is a set of tuples
$(\rid, m) \in \replicas \times M$, representing messages $m$ that are currently in transit towards replica $\rid$;
and an \emph{event trace} $\Gamma$ that is a list of events $(\rid, i, o)$.
Since events $(\rid, i, o)$ abstractly describe what computation took place at replica $\rid$,
the event trace captures the execution of the system so far: if $(\rid, i, o)$ is the $n$th element
of $\Gamma$, then replica $\rid$ transitions according to input $i$,
producing output $o$ at the $n$th transition of the execution. 
For example, if $\rid$ (chosen non-deterministically) makes a transition $s \xstepsto{(\rid, \bot, \send[m])} s'$,
the event trace $\Gamma$ takes a step by appending $(\rid, \bot, \send[m])$.
Initially,  $\Sigma^{0} = \lambda (\rid : \replicas)\, .\, s^{0}$ maps each replica to the
same initial state, and both $\Gamma^{0}$ and $\beta^{0}$ are empty.

Notably, since individual replicas output messages as $\send[m]$,
the choice of recipients is up to the system-level semantics.
In a similar vein, the choice of when to \emph{deliver} an in-flight message is also delegated to the system-level semantics.
For the rest of the paper, we assume messages are \emph{broadcast}, since 
broadcast can always be simulated by an appropriate sequence of point-to-point messages.
In other words, transitions $s \xstepsto{(\rid, i, \send[m])} s'$ indicate that $m$
should be broadcast to all other replicas $\rid' (\neq \rid)$.
To model this behavior, we use the following ``broadcast'' function:
\begin{gather*}
  \mathsf{bcast}(\rid, m)(\beta) = \beta \cup \{(\rid', m) \mid  \rid' \in \replicas \land \rid' \neq \rid\}.
\end{gather*}
For example, if $\replicas = \{\rid_1, \rid_2, \rid_3, \rid_4\}$,
then $\mathsf{bcast}(\rid_3, m)$ creates a packet
$\{(\rid_{1}, m), (\rid_{2}, m), (\rid_{4}, m)\}$
of $3$ copies of $m$ tagged with destinations $\rid_1, \rid_2, \rid_4$.
The packet is then pushed into the network state $\beta$ on behalf of $\rid_{3}$.
Because clients of a CRDT do not interact with the internal network itself,
we render all message-passing transitions as ``silent'' transitions $\quietleadsto$.

For the most part, we only care about configurations $\config{\Gamma}{\Sigma}{\beta}$
where  $\Gamma$, $\Sigma$, and $\beta$ ``agree'' with each other in some sense,
which we define as a ``well-formedness'' condition.

\begin{definition}[Well-Formed Configurations]\label{def:well-formed}
  Given an initial configuration $\mathtt{init}$, we say a configuration $\config{\Gamma}{\Sigma}{\beta}$
  is \emph{well-formed} if there is an execution from $\mathtt{init}$ to $\config{\Gamma}{\Sigma}{\beta}$,
  i.e., a sequence
    \[ \mathtt{init} \xleadsto{\alpha_1} \config{\Gamma_{1}}{\Sigma_{1}}{\beta_{1}} \cdots
    \xleadsto{\alpha_{k}} \config{\Gamma}{\Sigma}{\beta}. \]
\end{definition}

\subsection{Op-Based CRDT Semantics}\label{sec:opsem}

\begin{assumption}
  In this section, we fix the sets $S$, $\mathtt{Op}$, $M$, $Q$, and $V$ of
  resp. states, operations, messages, queries, and values. We assume a distinguished
  initial state $s^{0} \in S$, and $\prmap$, $\emap$, and $\qmap$ as in \cref{fig:op-based-object},
  and an ordered set $\replicas$ of replica identifiers. We assume messages $m \in M$ contain a unique identifier and timestamp,
  so they are distinct.
\end{assumption}

\paragraph{(Causal Order Delivery)}
As noted in \cref{subsec:op-based-crdts}, an op-based CRDT is built on top of a causal broadcast protocol
which implements Lamport's \emph{happens-before} in order to ensure that 
updates are delivered to other replicas in an order consistent with causality, while allowing the concurrent
messages to be delivered in any order, since they commute. Similarly to \citet{lamport-clocks}, we
define the \emph{happens-before} (or \emph{causal}) relation $\hb$ on all pairs of events in an event trace $\Gamma$. 
\begin{enumerate}
  \item \emph{(Program Order).} If $(\rid, i, o)$ occurs before $(\rid, i', o')$ in $\Gamma$,
        then replica $\rid$ executed event $(\rid, i, o)$ before event $(\rid, i', o')$, and so $(\rid, i, o) \hb (\rid, i', o')$;
  \item if $(\rid, i, \send[m])$ and $(\rid', \dlvr[m], o')$ are events, then $(\rid, i, \send[m]) \hb (\rid', \dlvr[m], o')$;
  \item if $(\rid, i, o) \hb (\rid'', i'', o'')$ and $(\rid'', i'', o'') \hb (\rid', i', o')$,
        then $(\rid, i, o) \hb (\rid', i', o')$.
\end{enumerate}

There are known, practical methods of ensuring message ordering characterizes causality,
such as including in message metadata timestamps in the form of a 
\emph{vector clock}~\cite{mattern-vector-time,birman-lightweight-cbcast}.
For example, let $\mathsf{vc}$ map a message to its vector clock,
and consider any execution of a distributed system.
Then there is an irreflexive partial order relation $\prec$ on pairs $\mathsf{vc}(m'), \mathsf{vc}(m)$ 
which completely characterize Lamport's happens-before relation $\hb$ 
in the sense that
\[ \mathsf{vc}(m') \prec \mathsf{vc}(m) \iff (\rid, i, \send[m']) \hb (\rid', i', \send[m]), \]
where $(\rid, i, \send[m'])$ and $(\rid', i', \send[m])$ are events in the execution.
Since the messages themselves contain the vector clock as metadata, the 
$\prec$ relation induces a partial order on $M$. Thus, we assume $(M, \prec)$ is a partial order which characterizes $\hb$
without loss of generality, and we write $m' \prec m$ to mean that $m'$ is a causal predecessor of $m$, in the above sense.
Moreover, we consider the \emph{concurrent} messages to be those which are incomparable with respect to $\prec$, i.e.,
  \[ \concurrent{m'}{m} \defeq \neg (m' \prec m) \land \neg (m \prec m'). \]

We now describe the replicas of an op-based CRDT, instantiating
\cref{def:replica} using the op-based CRDT specification in \cref{fig:op-based-object}
and the algorithm for the op-based update in \cref{fig:op-based-update} as a guide.

\begin{definition}[Op-Based Replica Semantics]\label{def:oplocal}
  An op-based replica $\rid \in \replicas$ is defined by the initial state $s^{0}$
  and the smallest labeled transition relation $\stepsto$ closed under the rules
  in \cref{fig:op-local-rules}.
\end{definition}

\begin{figure}[h]
  \scalebox{0.88}{$%
  \begin{gathered}
    \inference{
      q \in \mathtt{Q} & v = \qmap(q, s)
    }{
      s \xstepsto{(\rid, \qry[q] , \resp[v])} s
    }
    \quad
    \inference{
      m \in M & s' = \emap(m, s)
    }{
      s \xstepsto{(\rid, \dlvr[m] , \bot)} s'
    } 
    \quad
    \inference{
      \op \in \mathtt{Op} & m = \prmap(\rid, \op, s) & s' = \emap(m, s)
    }{
      s \xstepsto{(\rid, \upd[\op] , \send[m])} s'
    }
  \end{gathered}%
  $}
  \caption{Op-based replica Semantics. Notice that update commands $\upd[\op]$ first invoke
  $\prmap$ to generate a message $m$, then apply the message locally with $\emap$, c.f. \cref{fig:op-based-update}.
  }\label{fig:op-local-rules}
\end{figure}

\begin{definition}[Op-Based System Semantics]\label{def:opglobal}
  Given a set of op-based replicas $\replicas$ (of the same op-based object),
  the system semantics of an op-based CRDT
  are defined by the initial configuration
  $\mathtt{init} = \config{\varepsilon}{\Sigma^{0}}{\varnothing}$
  and the smallest labeled transition relation $\leadsto$ on
  configurations closed under the rules in \cref{fig:op-global-rules}.
\end{definition}

\begin{figure}[h]
  \scalebox{0.88}{
  $\begin{gathered}
    \inference[\textsf{OpUpdate}]{
      \rid \in \replicas & \op \in \mathsf{Op} & s = \Sigma(\rid) & e = (\rid, \upd[\op],\send[m]) &
      s \xstepsto{ e } s'
    }{
      \config{\Gamma}{\Sigma}{\beta}
      \xleadsto{(\rid, \upd[\op], \bot)}
      \config{\Gamma \cdot e}{\Sigma [\rid \mapsto s']}{\bcast(\rid, m)(\beta)}
    }\label{rule:OpUpdate}
    \\[1ex]
    \inference[\textsf{OpQuery}]{
      \rid \in \replicas & q \in Q & s = \Sigma(\rid) & e = (\rid , \qry[q] , \resp[v]) &
      s \xstepsto{e} s
    }{
      \config{\Gamma}{\Sigma}{\beta}
      \xleadsto{~ e ~}
      \config{\Gamma \cdot e}{\Sigma}{\beta}
    }\label{rule:OpQuery}
    \\[1ex]
    \inference[\textsf{OpDeliver}]{
      (\rid, m) \in \beta & (\rid, \dlvr[m]) \in \mathsf{enabled}(\Gamma) &
      s = \Sigma(\rid) &
      e = (\rid, \dlvr[m], \bot) &
      s \xstepsto{e} s'
    }{
      \config{\Gamma}{\Sigma}{\beta}\;
      \quietleadsto\;
      \config{\Gamma \cdot e}{\Sigma[\rid \mapsto s']}{\beta \setminus \{(\rid, m)\}}
    }\label{rule:OpDeliver}
  \end{gathered}$
  }
  \caption{Op-based system semantics. We assume an initial configuration
  $\mathsf{init} = \config{\varepsilon}{\Sigma^{0}}{\varnothing}$.}\label{fig:op-global-rules}
\end{figure}

\Cref{fig:op-global-rules} gives the system-level transition system for op-based CRDTs.
To capture the fact that op-based CRDTs are built on top of a causal-broadcast mechanism,
we use the $\mathsf{enabled}(\Gamma)$ predicate in \hyperref[rule:OpDeliver]{\textsf{OpDeliver}}
to ensure that messages $m$ can only be delivered after all causally preceding $m' \prec m$ have been
delivered. \cref{def:causal-message-delivery} formalizes this as a safety condition, and
\cref{def:enabled} defines $\mathsf{enabled}(\Gamma)$.
\cref{prop:causal-order-enforcement} shows that this safety condition is implied by \cref{def:enabled}
in well-formed configurations $\config{\G}{\S}{\B}$. The proof proceeds by unfolding
definitions \cref{def:causal-message-delivery} and \cref{def:enabled}, and doing induction
on the execution implied by well-formedness.

\begin{definition}[Causal Delivery Order]\label{def:causal-message-delivery}
  Let $M$ be a set and $m, m' \in M$ a pair of messages. Let $(\rid, i, o)_{i}$ denote
  the $i$th event in $\Gamma$.
  We say $\Gamma$ satisfies
  \emph{causal delivery order} if, $\forall m,m' \in M$, and $\forall \rid \in \replicas$,
  \begin{gather}
    (m \prec m') \land ((\rid, \dlvr[m], \bot)_{i} \in \Gamma) \land ((\rid, \dlvr[m'], \bot)_{j} \in \Gamma)\\
    \implies \neg (\dlvr[m'] \textit{ ordered before } \dlvr[m] \textit{ in } \Gamma \, \textit{i.e., } i < j).
  \end{gather}
\end{definition}

\begin{definition}\label{def:enabled}
  Given an event trace $\Gamma$, we say $\dlvr[m]$ is \emph{enabled} at $\rid$
  and write $(\rid, \dlvr[m]) \in \mathsf{enabled}(\Gamma)$
  if and only if both the following hold:
  \begin{enumerate}
    \item $\forall o \in O$, $(\rid, \dlvr[m], o) \notin \Gamma$ (i.e., $m$ has not already been delivered);
    \item $\forall m' \in M$ such that $m' \prec m$,
            we have $(\rid, \dlvr[m'], \bot) \in \Gamma$ (i.e., causally preceding $m'$ have been delivered).
  \end{enumerate}
\end{definition}
Note that the second condition in \cref{def:enabled} enforces the causal delivery order.
It says that any message $m'$ that is a potential cause for $m$ must be delivered \emph{before}
$m$ is allowed to be delivered.

\begin{proposition}[Causal Order Enforcement]\label{prop:causal-order-enforcement}
  Suppose $\config{\Gamma}{\Sigma}{\beta}$ is a well-formed op-based configuration.
  Then $\Gamma$ satisfies causal delivery order.
\end{proposition}

Implicit in \cref{def:causal-message-delivery} is that if $m'$ and $m$ are concurrent, i.e. $\concurrent{m'}{m}$,
then there is no constraint on the order of $\dlvr[m]$, and $\dlvr[m']$.
Recall from \cref{fig:op-based-object} that the effects of messages $m$ and $m'$ should
commute if $m$ and $m'$ are concurrent, i.e., neither is a cause for the other.
We encode this assumption in our semantics as follows. Let $\config{\Gamma}{\Sigma}{\beta}$ a configuration and $\rid$ a replica.
If $\concurrent{m'}{m}$, then
  \begin{gather*}
    \Sigma(\rid) \xstepsto{(\rid, \dlvr[m], \bot)} s_{1}' \xstepsto{(\rid, \dlvr[m'] , \bot)} s_{1}''\,
    \textit{  and  }\;
    \Sigma(\rid) \xstepsto{(\rid, \dlvr[m'] , \bot) } s_{2}' \xstepsto{(\rid, \dlvr[m] , \bot)} s_{2}'' \\
    \implies s_{1}'' = s_{2}''.
  \end{gather*}




\subsection{State-Based CRDT Semantics}\label{sec:stsemantics}

\begin{assumption}\label{op-based-assumptions}
  In this section, we fix the sets $S$, $\mathsf{Op}$, $Q$, and $V$ of
  resp. states, operations, queries, and values. 
  We assume $(S, \sqcup)$ is a join-semilattice with a distinguished
  initial state $s^{0}$, and we assume
  functions $\upmap$, $\mmap$, and $\qmap$ as in \cref{fig:st-based-object}.
\end{assumption}

The semantics of state-based CRDTs are straightforward, since there
are very few constraints on the behavior of the network other than reliable delivery.
Like before, state-based replicas can be thought of as an instantiation of \cref{def:replica}.

\begin{definition}[State-Based Replica Semantics]
  \label{def:stlocal}
  A replica $\rid \in \replicas$ of a state-based CRDT is defined by
  the initial state $s^{0}$ and the smallest labeled transition relation $\stepsto$
  closed under the rules in \cref{fig:st-local-rules}.
\end{definition}

\begin{figure}[h]
  \scalebox{0.88}{$
  \begin{gathered}
    \inference{
      q \in Q & v = \qmap(q, s)
    }{
      s \xstepsto{(\rid, \qry[q] , \resp[v])} s
    }
    \quad
    \inference{
      s'' \in S & s' = \mmap(s, s'')
    }{
      s \xstepsto{(\rid , \dlvr[s'] , \bot)} s'
    }
    \quad
    \inference{}{
      s \xstepsto{(\rid, \bot  , \send[s])} s
    }
    \quad
    \inference{
      \op \in \mathtt{Op} & s' = \upmap(\op, s)
    }{
      s \xstepsto{ (\rid , \upd[\op] , \bot)} s'
    }
  \end{gathered}
  $}
  \caption{State-based replica semantics. Note that $\send[s]$ outputs arise from a
    distinct transition step, c.f., \cref{def:oplocal}
  }\label{fig:st-local-rules}
\end{figure}

The semantics of state-based CRDTs are quite similar to op-based CRDTs at a local level,
only now replicas transmit their state separately from their local updates,
and this is considered a distinct (but silent!) computational event.
The state is still transmitted via reliable broadcast, but now there are no restrictions the delivery order,
thanks to the join-semilattice structure of the state space.
This manifests in the\StSend rule, where instead of an $\mathsf{enabled}$ predicate,
we use the premise $(\rid, \dlvr[s''], -) \notin \Gamma$, which is just the first condition \cref{def:enabled}.
The system-wide semantics are otherwise a straightforward lifting of the rules in \cref{def:stlocal}
to global configurations $\config{\Gamma}{\Sigma}{\beta}$.

\begin{definition}[State-Based CRDT Systems]
  \label{def:stglobal}
  Given a set of state-based replicas $\replicas$ (of the same state-based object),
  the system semantics of the state-based CRDT
  are defined by the initial configuration
  $\mathtt{init} = \config{\varepsilon}{\Sigma^{0}}{\varnothing}$
  and the smallest labeled transition relation $\leadsto$ on
  configurations closed under the rules in \cref{fig:st-global-rules}.
\end{definition}

\begin{figure}
  \scalebox{0.88}{$
  \begin{gathered}
    \inference[\textsf{StUpdate}]{
      \rid \in \replicas & \op \in \mathtt{Op} & s = \Sigma(\rid) & e = (\rid, \upd[\op], \bot) &
      s \xstepsto{~ e ~} s'
    }{%
      \config{\Gamma}{\Sigma}{\beta}
      \xleadsto{~ e ~}
      \config{\Gamma \cdot e}{\Sigma[\rid \mapsto s]'}{\beta}
    }\label{rule:StUpdate}
    \\[1ex]
    \inference[\textsf{StQuery}]{
      \rid \in \replicas & q \in Q & s = \Sigma(\rid) & e = (\rid, \qry[q], \resp[v]) & s \xstepsto{e} s'
    }{%
    \config{\Gamma}{\Sigma}{\beta}
    \xleadsto{~ e ~}
    \config{\Gamma \cdot e}{\Sigma}{\beta}
    }\label{rule:StQuery}
    \\[1ex]
    \inference[\textsf{StSend}]{
      \rid \in \replicas &  s = \Sigma(\rid) & e = (\rid, \bot, \send[s]) &
      s \xstepsto{ e } s
    }{%
      \config{\Gamma}{\Sigma}{\beta} \;
      \quietleadsto \;
      \config{\Gamma \cdot e}{\Sigma}{\bcast(\rid, s)(\beta)}
    }\label{rule:StSend}
    \\[1ex]
    \inference[\textsf{StDeliver}]{
      (\rid, s'') \in \beta & (\rid, \dlvr[s''], -) \notin \Gamma & s = \Sigma(\rid) & e = (\rid, \dlvr[s''], \bot) &
      s \xstepsto{ e } s'
    }{%
      \config{\Gamma}{\Sigma}{\beta}\;
      \quietleadsto \;
      \config{\Gamma \cdot e}{\Sigma[\rid \mapsto s']}{\beta \setminus \{(\rid, s'')\}}
    }\label{rule:StDeliver}
  \end{gathered}
  $}
  \caption{State-based system semantics. We assume an initial
  configuration
  $\mathtt{init} = \config{\varepsilon}{\Sigma^{0}}{\varnothing}$.}\label{fig:st-global-rules}
\end{figure}

\section{CRDT Emulation as (Weak) Simulation}
\label{sec:emulation}

An \emph{emulator} is a piece of software that allows one system to behave as a different system.
In this sense, the translation algorithms of \citeauthor{shapiro-crdts} really are `emulators',
in that they allow an op-based CRDT to behave as a state-based CRDT, and vice versa. 
Since op-based and state-based CRDTs are also systems that are understood to be `equivalent'
in some sense, we formally argue here that they are `behaviorally equivalent' in terms of
weak simulation. \cref{def:emulation} essentially defines the proof obligation for this section.

\begin{definition}[CRDT Emulation]\label{def:emulation}
  We say a mapping $\GG$ is a \emph{CRDT emulation} (or a \emph{CRDT emulator}) if, 
  given any \emph{host} CRDT object ${\OO^{\host}_{\kappa}}$ with system semantics 
  ${(\mathsf{Config}_{\kappa}, \initstate_{\kappa}, \leadsto_{\kappa})}$,
  there is a \emph{guest} CRDT object ${\OO^{\guest}_{\lambda}}$ with semantics
  $(\mathsf{Config}_{\lambda}, \initstate_{\lambda}, \leadsto_{\lambda})$
  such that:
  \begin{itemize}
    \item ${(\mathsf{Config}_{\kappa}, \initstate_{\kappa}, \leadsto_{\kappa})} \xmapsto{\GG} {(\mathsf{Config}_{\lambda}, \initstate_{\lambda}, \leadsto_{\lambda})}$; and
    \item the initial configurations weakly simulate each other, e.g., 
      $\initstate_{\kappa} \wsimulatedby \wsimulates \initstate_{\lambda}$.
  \end{itemize}
\end{definition}

In words, we formalize CRDT emulation as two components: (i) the \emph{map} $\GG$
which translates the CRDT objects; and (2) a pair of weak simulations $(\RR_1, \RR_2)$.
The relation $\RR_1$ shows the guest system weakly simulates the host system, and $\RR_2$ shows
the host system weakly simulates the guest system. Thus, a single `emulation' $\GG$ forms two proof obligations.
Recall that the emulation algorithms of \citeauthor{shapiro-crdts} offer
two such emulations: one from op-based CRDTs to state-based CRDTs, and one from state-based CRDTs to op-based CRDTS.
Here we focus on the emulation $\GG$ from \emph{\textbf{op-based}} CRDTs to \emph{\textbf{state-based}} CRDTs,
which we consider to be the more difficult and interesting of the two. 
It is also a good exemplar of our techniques. The emulation from state-based to op-based
CRDTs largely follows the same pattern, and can be found in \cref{appendix:other-emulation}.

\Cref{sec:op-to-state-emulation} describes the candidate mapping $\mathcal{G}$ used by the op-based to state-based direction, 
and \Cref{sec:op-to-state-details} describes the two weak simulations that make $\mathcal{G}$ an emulation. 
In \Cref{sec:preserved-properties} we make good on our earlier assertion by giving an example of transferring
a \emph{strong convergence} property, and also describing a subset of properties that transfer `for free'.
\Cref{subsec:refinement} concludes with a comparison of our techniques to some other methods which have appeared in the literature.

\subsection{The Mapping $\mathcal{G}$ from Op-Based to State-Based CRDTs}\label{sec:op-to-state-emulation}

\begin{assumption}
  For \cref{sec:op-to-state-emulation,sec:op-to-state-emulation}
  we assume an op-based CRDT object with its corresponding sets $S, \mathtt{Op}, M, Q, V$, related functions,
  and semantics as in \cref{fig:op-global-rules}. We further assume the set of messages
  is a partial order $(M, \prec)$ as explained in \cref{sec:opsem}. For notational simplicity,
  we at times apply messages $m$ as functions, e.g., $(m' \circ m)(s) \defeq \emap(m', \emap(m, s))$.
\end{assumption}

The mapping $\GG$ we consider is inspired by the original algorithm
described by \citet{shapiro-crdts}, and we direct interested readers there to see it in its original
presentation. Note that in both \citet{shapiro-crdts} and here, there is an underlying
assumption that op-based CRDTs are built on top of reliable causal broadcast.

In simple terms, op-based-to-state-based emulation relies on
what has classically been described as the \emph{semantic characterization of state machines},
which says that the current state of a state machine is defined only by
the sequence of requests it processes, and by nothing else~\cite{schneider-state-machine-approach}.
Replicas are indeed state machines, so
instead of maintaining the current op-based state $\opsem{s}$, a replica can instead maintain
the initial state $\opsem{s^{0}}$ and the set of messages \emph{that led to} $\opsem{s}$, i.e., an
ordered set of messages $\{m_{1}, \dots, m_{k-1}, m_{k}\}$ such that
\[ \opsem{s} = (m_k \circ m_{k-1} \circ \cdots \circ m_1)(\opsem{s^{0}}) \defeq \emap(m_k, \emap(m_{k-1}, ... \emap(m_1, \opsem{s^{0}}) ...)) .\]

The fact that each set of messages $\{m_1, \dots, m_k\}$
is an element of the join-semilattice $(\Pow(M), \cup)$ (where $\cup$ is set-theoretic union)
is even better --- there is a representation of each op-based state $\opsem{s}$ 
as a state-based state $\stsem{\{m_1, \dots, m_k\}}$.
We just need each replica to have access to: (1) the initial state $\opsem{s^{0}}$, and
(2) a function $\interp : \Pow(M) \to S$ that can interpret the set of messages $\stsem{\{m_1, \dots, m_k\}}$ into $\opsem{s}$. 
The first point is trivial --- we can always assume replicas have access to some additional metadata. 
For the second point, $\interp$ is roughly defined by taking a message set $\stsem{H} \in \Pow(M)$,
choosing some linear ordering $\langle m_{1}, \dots, m_{|\stsem{H}|} \rangle$ of $\stsem{H}$ that is consistent with 
the casual order $\prec$, then computing $(m_{|\stsem{H}|} \circ \cdots \circ m_1)(\opsem{s^{0}})$.

It turns out that since $\prec$ is an irreflexive, transitive relation that preserves and reflects causal order, 
we can use the fact that the effects of concurrent messages commute (\cref{fig:op-based-object}) to define $\interp$ recursively.
The key idea is to define $\mathsf{Max}(\stsem{H}) \defeq \{m \in \stsem{H} \mid \nexists m' \in \stsem{H} \text{ s.t. } m \prec m'\}$,
which is a set of `maximal' messages, with respect to $\prec$.
Critically, each pair of messages $m,m' \in \mathsf{Max}(\stsem{H})$ are concurrent, i.e., $\concurrent{m}{m'},$
so we can `peel off' a message $m \in \mathsf{Max}(\stsem{H})$ and recurse:
\begin{gather}
  \interp(\stsem{H}) \defeq
    \begin{cases}
      \opsem{s^{0}} & \text{if $H = \varnothing$} \\
      \opemap(\interp(\stsem{H} \setminus \{m\}), m) & \text{if $m \in \mathsf{Max}(\stsem{H})$}. \label{interp-def}
    \end{cases}
\end{gather}
This function is well-defined, and terminates whenever $\stH$ is finite (which is the case).
It essentially recursively picks an arbitrary order of the messages in $\stH$, but in such a
way as to be consistent with the order $\prec$. The proof follows from the fact that
each $m \in \mathsf{Max}(\stH)$ must be removed before the recursion can proceed with
the $\interp(\mathsf{Max}(\stH \setminus \mathsf{Max}(\stH)))$ call.

With these ingredients in place, the candidate emulation map $\GG$ is the one that takes
an op-based CRDT object as in \cref{fig:op-based-object} and constructs the state-based CRDT object shown in \Cref{fig:state-based-guest}.
We already know that $(\Pow(M), \cup)$ is a join-semilattice, so it remains to check that $\stupmap$ is inflationary ---
but it clearly is, since it's defined in terms of $\cup$.
\begin{figure*}[h!]
  \begin{align*}
        & \textbf{Parameters :} \\
        & \quad \Pow(M) : \text{states,}\; \mathtt{Op} : \text{operations,}\; 
        Q : \text{queries},\; V : \text{values,}\; \\
        & \quad  \varnothing \in \Pow(M) : \text{ initial state } \\
        & \textbf{ Functions :}\\
        & \quad \stupmap(\rid, \mathtt{op}, \stH) \defeq \stH \cup \{\opprmap(\rid, \op, \interp(\stH))\}\\
        & \quad \stmmap(\stH, \stsem{H'}) \overset{\mathrm{def}}{=}  \stsem{H} \cup \stsem{H'} \\
        & \quad \stqmap(q, \stsem{H}) \overset{\mathrm{def}}{=} \opqmap(q, \interp(\stsem{H}))
  \end{align*}
  \caption{The state-based guest CRDT constructed by $\GG$ from a given op-based host CRDT.}
  \label{fig:state-based-guest}
\end{figure*}

\subsection{The Weak Simulations}\label{sec:op-to-state-details}

We now proceed with defining the pair of weak simulations that make $\mathcal{G}$ an emulation in the sense of \Cref{def:emulation}.

\begin{notation}
  Let $\mathsf{Sent}(\G)$ to denote the set of `sent messages`, that is the messages with a corresponding $(\rid, i, \send[m])$ event
  in $\Gamma$. Likewise, let $\mathsf{Delivered}(\rid, \opsem{\G})$ be the set of `delivered' messages at
  replica $\opsem{\rid}$, i.e., the set of messages $\{m_1, \dots, m_k\}$ which $\opsem{\rid}$ either delivered with a
  $\dlvr[m]$ event, or generated with $\opsem{\prmap}$ and then consumed via $\opemap$ during an $\upd[\opsf]$ event.
  We use $(m)^{\darrow \G}$ to denote the set \emph{downward closure} of $m$ in the event trace $\G$. That is,
  $(m)^{\darrow \G} \defeq \{m' \in \mathsf{Sent}(\G) \mid  m' \prec m \} \cup \{m\}$.
  We write $\stleadsto{\subG}$ to denote the fact that the state-based CRDT system is a result of $\GG$.
\end{notation}

We assume all configurations $\opconfigA$, $\stconfigA$ are \emph{well-formed}
(\cref{def:well-formed}), and we fix a set of replicas $\{ \rid_1, \dots, \rid_{n}\} = \replicas$.
For space reasons, we sketch proofs by showing archetypical cases since
the omitted cases follow a similar pattern.
Complete proofs can be found in \cref{appendix:proofs}.

\subsubsection{The Guest CRDT Weakly Simulates the Host CRDT}
The weak simulation of $\opleadsto$ by $\stleadsto\subG$ is rather straightforward, but with one technical
hurdle that forms the heart of \cref{ex:non-bisim}, and thus of our simulation arguments.
The crucial detail is that whenever an op-based replica $\opsem{r}$ executes a $\opsem{\dlvr[m]}$
event, the state-based replica $\stsem{r}$ needs to pick a \emph{set} of messages $\stsem{H}$
whose merge (i.e., $\stsem{\dlvr[H]}$) \emph{simulates} the $\opsem{\dlvr[m]}$ deliver event.
This means our weak simulation needs to guarantee that a ``small enough'' $\stsem{H}$ is always available.
The right choice of $\stsem{H}$ is actually the down set $(m)^{\darrow \opsem{\G}}$, which
actually \emph{is} available so long as we always match \OpUpdate with \StUpdate followed by \StSend (broadcast).
The state-based system then broadcasts just as often as the op-based system, ensuring there is always a `small enough' state.

\begin{figure*}[h]
  \begin{subfigure}[c]{\textwidth}
    \centering
    \begin{align*}
        \RR_{1} = \{ & (\opconfig{\G}{\S}{\B} , \stconfig{\G}{\S}{\B}) \mid \\
                    & \textit{ --- configurations agree on messages, and replicas agree on delivered/merged messages} \\
                    & \quad (\mathsf{Sent}(\opsem{\G}) = \textstyle \bigcup \mathsf{Sent}(\stsem{\G}))
                      \land
                      \forall \rid \in \replicas\, :\, \mathsf{Delivered}(\rid, \opsem{\G}) = \stsem{\S}(\rid) \\
                    & \textit{ --- replicas agree on their local states} \\
                    & \land \quad \forall \rid \in \replicas\, :\, \opsem{\S}(\rid) = \interp(\stsem{\S}(\rid)) \\
                    & \textit{--- every pending message in the op-system has a pending downset in the state-system} \\
                    & \land \quad \forall (\rid, m) \in \opsem{\B} \implies (\rid, (m)^{\darrow \opsem{\G}}) \in \stsem{\B}
                  \}.
    \end{align*}
    \caption{Weak simulation of the op-based host CRDT by the state-based guest CRDT.}\label{fig:R1-sim}
  \end{subfigure}
  \\
  \begin{subfigure}[c]{\textwidth}
    \centering
    \begin{align*}
      \RR_{2} = \{ & (\stconfig{\Gamma}{\Sigma}{\beta} , \opconfig{\Gamma}{\Sigma}{\beta}) \mid \\
                   & \textit{ --- configurations agree on messages, and replicas agree on delivered/merged messages} \\
                   & \quad (\textstyle \bigcup \mathsf{Sent}(\stsem{\Gamma}) = \mathsf{Sent}(\opsem{\Gamma}))
                     \land 
                     \forall \rid \in \replicas\, :\,  \stsem{\Sigma}(\rid) = \mathsf{Delivered}(\rid, \opsem{\Gamma}) \\
                   & \textit{ --- replicas agree on their local states} \\
                   & \land \quad \forall \rid \in \replicas\, :\, \interp(\stsem{\Sigma}(\rid)) = \opsem{\Sigma}(\rid)\\
                   & \textit{--- every merge is a merge of a downset, and op-based replica can simulate that merge} \\
                   & \land \quad \forall (\rid, \stsem{H}) \in \stsem{\beta} \,:\, \exists U \in \mathsf{deliverable}_{\opsem{\Gamma}}(\rid) \,:\, \stsem{H} = (U)^{\darrow \opsem{\Gamma}}
                \}.
    \end{align*}
    \caption{Weak simulation of the state-based guest CRDT by the op-based host CRDT.}\label{fig:R2-sim}
  \end{subfigure}
  \caption{The relation $\RR_{1} \subseteq \opsem{\mathsf{Config}} \times \stsem{\mathsf{Config}}$ is a weak simulation
          of op-based \emph{host} CRDT is simulated by the state-based \emph{guest} CRDT. Dually, the
          relation $\RR_{2} \subseteq \stsem{\mathsf{Config}} \times \opsem{\mathsf{Config}}$ is a weak simulation
          of the state-based \emph{guest} CRDT by the op-based \emph{host}. These relations,
          while very similar to each other, are not
          converses of each other --- indeed, they cannot be (as \cref{ex:non-bisim} shows).
  }
  \label{fig:R1-R2-sim}
\end{figure*}

\begin{theorem}\label{thm:main1}
  $\mathcal{R}_{1}$ (\Cref{fig:R1-R2-sim}) is a weak simulation of $\opleadsto$ by $\stleadsto{\subG}$
  that relates the initial configurations $\opconfig{\varepsilon}{\S^{0}}{\varnothing}$
  and $\stconfig{\varepsilon}{\S^{0}}{\varnothing}$.
\end{theorem}

$\RR_{1}(\opconfig{\varepsilon}{\S^{0}}{\varnothing}, \stconfig{\varepsilon}{\S^{0}}{\varnothing})$
is immediate. We only need to show that $\RR_{1}$ is indeed a weak simulation.
We prove the following archetypical case. The other cases follow a similar pattern.

\begin{proof}[Sketch of \Cref{thm:main1}]
  Assume the hypothesis:
    \[ \RR_1(\opconfigA, \stconfigB) \textit{ and for any $\alpha$, we have} \opconfigA \xopleadsto{\alpha} \opconfigB.\]
    The archetypical case is
    \begin{itemize}
      \item (\OpDeliver). Then $\alpha = \quiet{\tau}$, and we have $\opconfigA \xquietleadsto{\tau} \opconfigB$,
            where
            \[ \opsem{\G'} = \opsem{\G \cdot (\rid, \dlvr[m], \bot)}, \quad 
               \opsem{\S'} = \opsem{\S[\rid \mapsto \emap(m, \S(\rid))]}, \quad
               \opsem{\B'} = \opsem{\B \rmv{ (\rid, m) }},
            \]
            which means $\opsem{\rid}$ delivered an enabled message $(\rid, m) \in \opsem{\B}$.
            The hypothesis implies $\exists (\rid, (m)^{\darrow \opsem{\G}}) \in \stsem{\B}$. Therefore
            we can invoke\StDeliver and match with
              \begin{gather*}
                \stconfigA \quietleadsto \subG \stconfig{\G'}{\S'}{\B'}, \\
                \stsem{\G'} = \stsem{\G \cdot (\rid, \dlvr[(m)^{\darrow \opsem{\G}}], \bot)}, \quad
                \stsem{\S'} = \stsem{\S[\rid \mapsto \S(\rid) \cup (m)^{\darrow \opsem{\G}}]}, \quad
                \stsem{\B'} = \stsem{\B \rmv{ (\rid, (m)^{\darrow \opsem{\G}}) }}.
              \end{gather*}
            To finish, we need to show $\RR_1(\opconfigB, \stconfigB)$ by checking the conditions in \cref{fig:R1-R2-sim}.
            Since no messages were sent, the last condition is immediate from the hypothesis. We focus
            on the changes that happened at replicas $\opsem{\rid}$, $\stsem{\rid}$. 
            Since $\opsem{\rid}$ delivered $m$, by \cref{prop:causal-order-enforcement}, 
            all causally preceding messages (i.e., $m' \prec m$) had to have already been delivered to at replica $\opsem{\rid}$.
            But since $\mathsf{Delivered}(\rid, \opsem{\G}) = \stsem{\Sigma(\rid)}$ (by hypothesis), all such $m' \in \stsem{\Sigma(\rid)}$
            as well. Thus, 
              \[ \stsem{\S'(\rid)} = \stsem{\S(\rid)} \cup (m)^{\darrow \opsem{\G}} = 
                \stsem{\S(\rid)} \cup \{m\} = \mathsf{Delivered}(\rid, \opsem{\G'}),\]
            which establishes the first condition in $\RR_1$. It remains to show that
            $\opsem{\S'}(\opsem{\rid}) = \interp(\stsem{\S'}(\stsem{\rid}))$. Indeed, since $m$
            satisfies either $\concurrent{m'}{m}$ or $m' \prec m$ for each $m' \in \stsem{\S(\rid)}$,
            \cref{interp-def} implies key equality
              \[ \interp(\stsem{\S'(\rid)}) = \interp(\stsem{\S(\rid)} \cup \{m\}) = 
                \opemap(\interp(\stsem{\S(\rid)}), m) = \opsem{\S'(\rid)}.\]
    \end{itemize}
\end{proof}

\subsubsection{The Host CRDT Weakly Simulates the Guest CRDT}

The weak simulation of $\stleadsto\subG$ by $\opleadsto$ proceeds similarly.
The only real technical difficulty is maintaining the invariant:
if the state-based replica $\stsem{\rid}$ invokes $\stmmap$, the op-based system can pick a sequence
of messages $(\opsem{\rid}, m_1), \dots, (\opsem{\rid}, m_k) \in \opsem{\B}$ whose sequential delivery simulates that merge.
In principle, this is easy, since the states $\stsem{H}$ are sets of messages, though one needs to keep
the deliverability condition of\OpDeliver in mind. To that end, we give the auxiliary \cref{def:deliverable-set}.

\begin{definition}\label{def:deliverable-set}
  In an op-based configuration $\opconfigA$, we say a
  subset of messages
  \[ U = \{m_1, \dots, m_k \} \subseteq M \]
  is \emph{deliverable at a replica} $\rid$ (and write $U \in \mathsf{deliverable}_{\opsem{\Gamma}}(\rid)$)
  if for all $j \in 1..k$,
  \begin{enumerate}
    \item $(\rid, m_j) \in \opsem{\beta}$ and
    \item $(\rid, \dlvr[m_j]) \in \mathsf{enabled}(\opsem{\Gamma} \cdot (\rid, \dlvr[m_1], \bot) \cdots (\rid, \dlvr[m_{j-1}], \bot))$.
  \end{enumerate}
\end{definition}

\begin{theorem}\label{thm:main2}
  $\mathcal{R}_{2}$ (shown in \cref{fig:R1-R2-sim}) is a weak simulation of $\stleadsto\subG$ by $\opleadsto$
  which relates the initial configurations $\stconfig{\varepsilon}{\S^{0}}{\varnothing}$
  and $\opconfig{\varepsilon}{\S^{0}}{\varnothing}$.
\end{theorem}

\begin{proof}[Sketch of \cref{thm:main2}]
  Assume as our hypothesis:
  \[ \RR_{2}(\stconfigA, \opconfigA)\; \textit{and for any $\alpha$, we have}\; \stconfigA \xstleadsto{\alpha} \subG \stconfigB. \]
  The interesting case is,
  \begin{itemize}
    \item (\StDeliver). Then $\alpha = \quiet{\tau}$ and we have $\stconfigA \quietleadsto \subG \stconfigB$ where
          \[  \stsem{\G'} = \stsem{\G \cdot (\rid, \dlvr[H], \bot)}, \quad
              \stsem{\S'} = \stsem{\S[ \rid \mapsto \mmap(\S(\rid), H)]}, \quad
              \stsem{\B'} = \stsem{\B \rmv{ (\rid, H) }},
          \]
          which means a replica $\rid$ merged a $(\rid,\stsem{H}) \in \stsem{\B}$ that was sent by some other replica.
          If $\stsem{\S(\rid)} \cup \stsem{H} = \stsem{\S(\rid)}$, we match by the reflexive step $\opconfigA \quietleadstostar \opconfigA$,
          and there is nothing to show, so we suppose there is a number $k > 0$, and messages $m_1, \dots, m_k \notin \stsem{\S(\rid)},$
          so that $\stsem{\S(\rid)} \cup \stsem{H} = \stsem{\S(\rid)} \cup \setOf{m_1, \dots, m_k}$.

          \paragraph{(Claim)}
          We can pick as $U \in \mathsf{deliverable}_{\opsem{\G}}(\rid)$ the set $\setOf{m_1, m_2, \dots, m_k}$,
          therefore deliver the sequence of messages $m_1 m_2 \cdots m_k$ at replica $\opsem{\rid}$ so that
            \begin{gather*}
              \opsem{\S(\rid)} \xopstepsto{(\rid , \dlvr[m_1] , \bot)} \opsem{s_1} 
                \xopstepsto{(\rid , \dlvr[m_2] , \bot)} \opsem{s_2} \cdots \xopstepsto{(\rid , \dlvr[m_k] , \bot)} \opsem{s_k}
            \end{gather*}
          and thereby match with $k$ applications of (\OpDeliver) with respect to replica $\opsem{\rid}$ giving:
          \begin{gather}\label{eq:deliveries1}
            \begin{gathered}
              \opconfigA \quietleadstostar \opconfigB,\\
              \opsem{\G'} = \opsem{\G} \cdot \opsem{\langle (\rid, \dlvr[m_1], \bot), \dots, (\rid, \dlvr[m_k], \bot) \rangle},\\
              \opsem{\S'} = \opsem{\S[\rid \mapsto s_k]}, \quad
              \opsem{\B'} = \opsem{\B} \opsem{\rmv{(\rid, m_1), \dots, (\rid, m_k)}}.
            \end{gathered}
          \end{gather}
      
          \paragraph{(Proof of Claim)}
          We can partition $\stH$ into two sets: a set of delivered messages $D \subseteq \stsem{\S(\rid)}$,
          and a set $U = \{m_1, \dots, m_k\}$ of $k \geq 0$ undelivered message, i.e., $U \cap \stsem{\S(\rid)} = \varnothing$.
          We need to show that $\opsem{\rid}$ can pick and deliver $U$.
          
          From the hypothesis, $\exists W \in \mathsf{deliverable}_{\opsem{\G}}(\rid)$ such that $\stsem{H} = (W)^{\darrow \opsem{\G}}.$
          Since $(W)^{\darrow \opsem{\G}}$ is the downwards closure of $W$,
          we immediately have $W \subseteq (W)^{\darrow \opsem{\G}} = \stH$. Thus we can similarly partition $W$ into
          a set $D'$ of delivered messages,
          and a set $U'$ of undelivered messages. By construction, $U' \subseteq U$.
          From \cref{def:deliverable-set,def:enabled}, it must be the case that $D' = \varnothing$.
          So, $W = U'$. If we can show $U \subseteq W$, we are done. Let $m \in U$.
          Note that $m$ must live in $\beta(\opsem{\rid})$ somewhere, since if not, 
          then we have $m \in \bigcup \mathsf{Sent}(\stsem{\G})$ and $m \notin \mathsf{Sent}(\opsem{\G})$, which is a contradiction.
          Suppose $m \notin W$. In that case, $m$ must be either `below'
          all messages $W$, or `above' all messages in $W$, with respect to the order $\prec$. In the former case,
          we have a contradiction, since $W$ is a deliverable set. In the latter case, we can always extend $W$ with
          all the undelivered causal ancestors of $m$, and $W$ is still a deliverable set. It follows that $m \in W$, hence $U \subseteq W$.
          Therefore, we can deliver the messages $W = U = \{m_1, \dots, m_k\}$ (reindexing if necessary), and we have the claim.
          \stopcase

          It remains to show that $\RR_{2}(\stconfigB, \opconfigB)$. 
          For that, note the hypothesis $\stsem{\Sigma}(\rid) = \mathsf{Delivered}(\rid, \opsem{\Gamma})$ immediately yields,
              \begin{gather}
                \stsem{\S'(\rid)} = \stsem{\S(\rid)} \cup U  = \mathsf{Delivered}(\rid, \opsem{\G}) \cup U = \mathsf{Delivered}(\rid, \opsem{\G'}). \label{equal-deliveries}
              \end{gather}
          The only other interesting condition is showing that 
          replicas agree in their local states, where it suffices to prove $\opsem{s_k} = \interp(\stmmap(\stsem{\S(\rid)}, \stsem{H}))$.
            But this has to be the case: since $\setOf{m_1, \dots, m_k} = U \in \mathsf{deliverable}_{\opsem{\G}}(\rid)$ means each pair 
            $m_i, m_j \in U$ $(i < j)$ satisfies either $m_i \prec m_j$ or $\concurrent{m_i}{m_j}$,
            and \eqref{equal-deliveries} means $\opsem{\rid}$ and $\stsem{\rid}$ have delivered the same set of messages, we obtain
              \[ \opsem{s_k} = (m_k \circ \cdots \circ m_1)(\opsem{\S(\rid)}) = \interp(\stsem{\S(\rid)} \cup U) = \interp(\stmmap(\S(\rid), \stsem{H})).\]
  \end{itemize}
\end{proof}

\subsection{Consequences of CRDT Emulation as Simulation}\label{sec:preserved-properties}

\Cref{thm:main1,thm:main2} show that CRDTs related by an emulation $\GG$ are behaviorally related.
In this section, we discuss some consequences of this relationship, and --- in particular ---
demonstrate how we can transfer \emph{properties} of the behavior of a CRDT across this relationship.

\subsubsection{Trace Properties}

We first
recall the definition of a weak trace set. Let $(X, \Lambda, \stepsto)$
be a labeled transition system. The \emph{weak trace set} of $x \in X$
is defined as
\[
    \wtrace(x) \defeq \{ \langle \alpha_1, \dots, \alpha_k \rangle \in (\Lambda \setminus \{\quiet{\tau}\})^{\ast} \mid
        \exists x_1,\dots,x_k \in X\, :\, (x \xstepstostar{\alpha_1} x_1 \cdots \xstepstostar{\alpha_k} x_k)
    \}.
\]
In words, it is the set of all finite (and unbounded) observable behaviors starting from a state $x$.
Critically, weak simulations are \emph{sound} with respect to \emph{weak trace inclusion} in the following sense:
if $x \in X$ and $y \in Y$ are states, then
\[ x \wsimulatedby y \implies \wtrace(x) \subseteq \wtrace(y).\]
It follows that if both $x \wsimulatedby y$ and $y \wsimulatedby x$, then $\wtrace(x) = \wtrace(y)$,
which is called \emph{trace equivalence.} For us, traces are in terms of the labels of
our transition semantics in \cref{sec:coalg}:
  \[
    \Lambda \setminus \{\quiet{\tau}\} \defeq \{(\rid, \upd[\op], \bot) \mid \rid \in \replicas, \op \in \mathtt{Op}\}
    \cup
    \{(\rid, \qry[\op], \resp[v]) \mid \rid \in \replicas, q \in Q, v \in V\}.
  \]
Our results \cref{sec:emulation} thus imply the following two corollaries.

\begin{corollary}[CRDT Emulation Implies Trace Equivalence]\label{consequence1}
    Let $\OO_{\host}$ an op-based (state-based) CRDT with transition semantics $\leadsto_{\host}$
    and initial configuration $\initstate_{\host}$. Let $\OO_{\guest}$
    a state-based (op-based) CRDT, with corresponding semantics $\leadsto_{\guest}$ and initial
    configuration $\initstate_{\guest}$. Suppose $\GG$ is an emulation such that,
    $\OO_{\guest} = \GG(\OO_{host})$. Then $\initstate_{\guest} \wsimulatedby \wsimulates \initstate_{\host}$,
    and therefore
      \[ \wtrace(\initstate_{\host}) = \wtrace(\initstate_{\guest}).\]
\end{corollary}

\begin{corollary}[Emulation Preserves Weak Trace Properties]\label{consequence3}
    Suppose $P$ is a predicate on weak traces 
    (i.e., a $P \subseteq \TT_{\Lambda}$, where $\TT_{\Lambda}$ is the set of all weak traces).
    If $\OO_{\host}$ and  $\OO_{\guest}$ are two CRDTs related by an emulation $\GG$,
    then we have, for all configurations $\Cfg_{\host}$ and $\Cfg_{\guest}$,
    \[ \Cfg_{\host} \wsimulatedby \wsimulates \Cfg_{\guest} \implies 
    (\wtrace(\Cfg_{\host}) \subseteq P \iff \wtrace(\Cfg_{\guest}) \subseteq P),\]
    and in particular, the following is true: $\wtrace(\initstate_{\host}) \subseteq P \iff \wtrace(\initstate_{\guest}) \subseteq P.$
\end{corollary}

\cref{consequence1} implies op-based and state-based CRDTs are equivalent precisely
in terms of the traces of \emph{update events} $(\rid, \upd[\op], \bot)$,
and \emph{query} events $(\rid, \qry[q], \resp[v])$.
\cref{consequence3} characterizes a set of properties that can be safely verified on one kind of CRDT (e.g., op-based),
and then `transferred' to the other kind of CRDT (e.g., state-based) `for free' --- the \emph{weak trace properties}.
Implicit in this formalism is reliable causal broadcast, but so long as that holds, state-based and op-based
CRDTs really are interchangeable in terms of their trace behaviors.

\subsubsection{Transfer of Other Properties}

Since simulation is stronger than trace inclusion, one might wonder if we can `transfer'
stronger properties than the trace properties. 
The answer is `yes', though unlike the weak trace properties, it is not completely free.
We show how one might do this with an example using \emph{strong convergence}, described in \cref{subsec:crdt-background}.
Although it is known that any well-formed CRDT satisfies strong convergence~\citep{shapiro-crdts},
it affords a good minimal example of how to transfer more complex properties across an emulation.

\begin{example}[Transfer of Strong Convergence]\label{ex:transfer-sec}

    In general, CRDTs do not need to keep track of the set of operations that have already been
    applied. As such, for any particular object $\OO$, we will
    construct an \emph{augmented}, history-carrying object $\OO'$ that is otherwise equivalent.

    Suppose that $\OO$ is an op-based CRDT with sets $S, \mathsf{Op}, M, Q, V$, initial state $s^{0}$,
    and functions $\mathsf{prep}, \mathsf{effect}, \mathsf{query}$. (Without loss of generality, we will assume
    that no operation occurs more than once, since we can always attach unique identifiers if necessary.)
    We define $\OO'$ via $S' = S \times \mathcal{P}(\mathsf{Op})$, $M' = M \times \mathcal{P}(\mathsf{Op})$,
    and $V' = V \times \mathcal{P}(\mathsf{Op})$ (leaving $\mathsf{Op}$ and $Q$ unchanged), with initial state
    $s'^{0} = (s^{0}, \varnothing)$, and functions
    \begin{gather*}
        \mathsf{prep}'(\rid, \mathsf{op}, (s, -)) = (\mathsf{prep}(\rid, \mathsf{op}, s), \{\mathsf{op}\})
        \qquad 
	      \mathsf{effect}'((m, h'), (s, h)) = (\mathsf{effect}(m, s), h \cup h') \\
        \mathsf{query}'(q, (s, h)) = (\mathsf{query}(q, s), h).
    \end{gather*}
    If $\OO$ is a state-based CRDT, we proceed similarly, noting that $(\mathcal{P}(\mathsf{Op}), \cup)$
    is a join semilattice. Intuitively, $\OO'$ is a kind of product of $\OO$ together with a CRDT that simply records
    and reports the operations accrued at a replica. The key invariant to observe is that two replicas have
    applied the same set of updates \emph{if and only if} they report the same set of updates under $\mathsf{query}'$.
   
    We can now consider \emph{executions} of $\OO'$, following our development in \Cref{sec:coalg}. Notice that
    every execution of $\OO'$ can be reduced to a valid execution of $\OO$ by \emph{forgetting} all history-related
    augmentations; and every execution of $\OO$ can be \emph{extended} to a valid execution of $\OO'$ by including the
    extra history-tracking information. Thus, the augmented semantics of $\OO'$ is \emph{sound and complete} with
    respect to the semantics of $\OO$, so we lose nothing by working with the executions of $\OO'$.
    In particular, the pair of weak simulations obtained in \Cref{sec:op-to-state-details} between $\OO$ and
    $\mathcal{G}(\OO)$ extends to a pair of weak simulations between their history-carrying augmentations.

    We now show how strong convergence `transfers'. With the augmented semantics out of the way, the proof is direct.
    Suppose we have an object $\OO_{\host}$, and an emulation 
    $\GG$,
    so that $\OO_{\guest} = \GG(\OO_{\host})$. Let $\initstate_{\host}$ and $\initstate_{\guest}$ be the initial
    configurations. Suppose that $\OO_{\guest}$ satisfies strong convergence, i.e.,
    if $\config{\G'}{\S'}{\B'} \in \mathsf{Config}'$ is a well-formed configuration, then
    \begin{gather}
        \forall \rid, \rid' : q(\Sigma'(\rid)) = (v_1, h_1) \land \Sigma(\rid') = (v_2, h_2) \land h_1 = h_2 \implies v_1 = v_2
        \label{SC}
    \end{gather}
    
    Let $\Sigma'_{\host}$ be the replica states of some well-formed configuration $\CC_{\host}$, such that
    replicas $\rid$ and $\rid'$ have the same set of updates, i.e., 
    \begin{gather}
        q(\Sigma'_{\host}(\rid)) = (v_1, h_1) \land q(\Sigma'_{\host}(\rid') = (v_2, h_2)) \implies h_1 = h_2.
        \label{premise}
    \end{gather}
    Since $\GG$ is an emulation, we have $\mathsf{init}_{\host} \wsimulatedby \wsimulates \mathsf{init}_{\guest}$,
    and therefore there are replica states $\Sigma'_{\guest}$ from a well-formed configuration $\CC_{\guest}$ 
    such that $\CC_{\host} \wsimulatedby \CC_{\guest}$,
    which implies agreement on queries:
    \begin{gather}
        q(\Sigma'_{\host}(\rid)) = (v_1, h_1) = q(\Sigma'_{\guest}(\rid))
        \quad \textit{and} \quad
       q(\Sigma'_{\host}(\rid')) = (v_2, h_2) = q(\Sigma'_{\guest}(\rid'))
       \label{simresult}
    \end{gather}
    But $\OO_{\guest}$ satisfies strong convergence, and so by \eqref{premise} and \eqref{simresult}, 
    we can apply \eqref{SC} with $\Sigma' = \Sigma'_{\guest}$, thus obtaining $v_1 = v_2$, and therefore,
    \[ q(\Sigma'_{\host}(\rid)) = (v_1, h_1) = (v_2, h_2) = q(\Sigma'_{\host}(\rid')).\]\
    Since this works for any well-formed configuration of the host object, we conclude that $\OO_{\host}$
    satisfies strong convergence.
\end{example}

\subsubsection{What of Bisimulation?}
Recall in \cref{ex:non-bisim} that \emph{simulations} do not necessarily respect \emph{branching behavior},
unlike (weak) bisimulation.
More generally, lack of bisimulation can present issues when we are concerned with \emph{deadlock},
\emph{termination}, and arbitrary `looping' behavior.
Simulation is well known to be insensitive to these concepts.
In particular, if one only has a pair of simulations $(\RR_1, \RR_2)$ which relate non-deterministic programs $p_{x}$ and $p_{y}$,
then one cannot prove statements of the form: \textit{``if \textbf{\emph{every}} execution of $p_x$ terminates, then \textbf{\emph{every}} execution of $p_y$ terminates.''}
There may yet be other $p_y$-executions that do not terminate!
In this case, the best we can do is show that $p_x$ and $p_y$ `approximate' each other in the sense that:
if $p_x$ \emph{can} terminate, \emph{then so can} $p_y$.
Our results in \cref{sec:representation-independence} show, among other things, a mechanically verified proof of this fact.

That said, the lack of a weak bisimulation is somewhat artificial.
It turns out one can obtain a weak bisimulation by modifying the state-based CRDT to broadcast just as often as the
op-based CRDT.
\begin{theorem}\label{thm:bisim-possible}
    In \cref{def:stglobal}, delete the \StSend rule, and in place of \StUpdate, use the following rule:
    \begin{gather*}
    \inference[\textsf{StUpdBC}]{
      \rid \in \replicas & \op \in \mathtt{Op} & s = \Sigma(\rid) &
      s \xstepsto{ (\rid, \upd[\op] , \send[s']) } s'
    }{
      \config{\Gamma}{\Sigma}{\beta}
        \xleadsto{ (\rid , \upd[\op] , \bot) }
      \config{\Gamma \cdot (\rid, \upd[\op], \send[m])}{\Sigma [\rid \mapsto s']}{\bcast (\rid, s')(\beta)}
    }\label{rule:StUpdBC}
    \end{gather*}
    Then, if $\opsem{\OO}$ is an op-based CRDT, and $\stsem{\OO}$ is a state-based CRDT
    which are related by some emulation $\GG$ (in either direction), $\opsem{\initstate}$
    and $\stsem{\initstate}$ are weakly \emph{bisimilar}, i.e.,
    $\opsem{\initstate} \approx \stsem{\initstate}$.
\end{theorem}

\subsection{Comparison with Data Refinement}\label{subsec:refinement}

Our approach is largely inspired by general theories of bisimulation and coinduction~\citep{milner-ccs,jacobs-book},
which place a premium on the behavioral equivalence of systems.
However, our approach concerns not two completely different systems, but rather a given system
and an `emulating' system, and so our methods also closely aligns with \emph{data refinement},
which also uses formal simulation to verify if two related implementations of a datatype satisfy
the same \emph{specification}~\cite{lynch1988,lamport-TLA,burckhardt-rdts-svo}.
In data refinement, one typically assumes an \emph{abstract} implementation that satisfies the specification, and then shows
that the concrete implementation is \emph{simulated} by the abstract implementation (\citet{lynch1988} give a good introduction).
Since we characterize systems by their behaviors, i.e., \emph{traces} (\cref{subsec:simulations-background}),
and simulation implies \emph{trace inclusion}, correctness is formalized as follows.

Let $\OO_{\mathtt{conc}}$, and $\OO_{\mathtt{abs}}$ be (resp.) concrete and abstract implementations of 
some object and let $\FF$ denote a \emph{correctness property}, seen as a subset of possible traces.
If $y^{\initstate}_{\mathtt{conc}}$ and $x^{\initstate}_{\mathtt{abs}}$ denote the (resp.) concrete and abstract initial states,
then we say $\OO_{\mathtt{conc}}$ \emph{refines} $\OO_{\mathtt{abs}}$ if $x^{\initstate}_{\mathtt{abs}}$ simulates $y^{\initstate}_{\mathtt{conc}}$.
In other words, we have the proof obligation,
\begin{gather}
  y^{\initstate}_{\mathtt{conc}} \simulatedby x^{\initstate}_{\mathtt{abs}} 
  \quad \textit{and therefore}
  \quad \trace(y^{\initstate}_{\mathtt{conc}}) \subseteq \trace(x^{\initstate}_{\mathtt{abs}}) \subseteq \FF.
  \label{data-refinement}
\end{gather}
Trace inclusion means the implementation $\OO_{\mathtt{conc}}$ is correct.
One of the original presentations of the above approach was given in \citet{lynch1988}.

\citep{burckhardt-rdts-svo}'s approach modifies the above idea to work for CRDT verification.
For one, they do not assume the existence of the abstract implementation
$\OO_{\mathtt{abs}}$, but rather begin with the concrete implementation $\OO_{\mathtt{conc}}$,
described in terms of a state space, sets of messages and operations, and transition
functions $\mathtt{do}$ (execute operations), $\mathtt{send}$ and $\mathtt{receive}$ (message handling).
To show $\OO_{\mathtt{conc}}$ satisfies the specification $\FF$,
they \emph{construct} a set of abstract executions satisfying $\FF$ (as opposed to assuming $\OO_{\mathtt{abs}}$ is given). 
This is done with an \emph{abstraction function} $\mathtt{abs}$ that sends every concrete execution $C$ of $\OO_{\mathtt{conc}}$
to an abstract execution $A$, which is essentially a trace $\mathtt{do}$ events, operations, return values, and some additional metadata.
Finally, to show that $A$ satisfies $\FF$, they use \emph{replication-aware simulations}. We direct the reader
to \citet{burckhardt-rdts-svo} for details.

If one squints, one can see the resemblance of \citeauthor{burckhardt-rdts-svo}'s approach and \eqref{data-refinement}
if instead of assuming the existence of $\OO_{\mathtt{abs}}$ and $x^{\initstate}_{\mathtt{abs}}$,
we \emph{generate it} by the abstraction map $\mathtt{abs}$,
e.g.,
\begin{gather}
  y^{\initstate}_{\mathtt{conc}} \simulatedby \mathtt{abs}( y^{\initstate}_{\mathtt{conc}})
  \quad \textit{and therefore}
  \quad \trace(y^{\initstate}_{\mathtt{conc}}) \subseteq \trace(\mathtt{abs}( y^{\initstate}_{\mathtt{conc}})) \subseteq \FF.
\end{gather}

It is here that our approach aligns closest, if one imagines the emulation map $\GG$ as playing the same role of the abstraction
function $\mathtt{abs}$. But there are subtle semantic differences between \citeauthor{burckhardt-rdts-svo}'s approach and ours.
For one, we do not assume a specification of correct behavior, because our work relates ``concrete executions'' of one object
$\OO_{\kappa}$ to the \emph{concrete executions} of another object $\OO_{\lambda}$ (as related by $\GG$).
Data refinement is a tool to prove correctness, whereas we are concerned with \emph{behavioral equivalence}.
In full generality, emulation is not required to preserve
behavior, or satisfy the same specification in the sense of data refinement\footnote{Indeed, consider a Nintendo 64 emulator running on a Windows PC. 
In this case, which is the concrete implementation, and which is the abstract
implementation? Should they both satisfy the same spec? We think data refinement does not capture the potentially complex relationship here.},
and it is not clear which object is the `concrete' implementation, and which is the `abstract' implementation,
since it just so happens that when it comes to CRDT emulation, behaviors \emph{are} expected to align.

In any case, the methods of \citeauthor{burckhardt-rdts-svo} (and data refinement more generally) can embed into our model.
Recall \cref{ex:transfer-sec}, where we transported via emulation a strong convergence property from one object to another.
If one models the strong convergence property as a \emph{specification} $\FF_{\mathtt{SC}}$,
and we view $\OO_{\host}$ as the `concrete' implementation, and $\OO_{\guest}$ as the `abstract' implementation,
then the emulation $\GG$ plays the role of the `abstraction function', since it 
completely determines a function $\mathtt{abs}_{\GG}$ which sends the `concrete' executions of $\OO_{\host}$
to the `abstract' executions of $\OO_{\guest}$. Since $\OO_{\guest}$ is assumed to satisfy $\FF_{\mathtt{SC}}$,
we obtain that $\OO_{\host}$ satisfies $\FF_{\mathtt{SC}}$ as well, in a manner similar to \citeauthor{burckhardt-rdts-svo}.

\def\tri{\triangleright}
\def\Gk{\G_{\kappa}}
\def\Gl{\G_{\lambda}}
\def\Sk{\S_{\kappa}}
\def\Sl{\S_{\lambda}}
\def\Bk{\B_{\kappa}}
\def\Bl{\B_{\lambda}}
\def\Gammak{\Gamma_{\kappa}}
\def\Gammal{\Gamma_{\lambda}}
\def\Sigmak{\Sigma_{\kappa}}
\def\Sigmal{\Sigma_{\lambda}}
\def\stepsk{\stepsto_{\kappa}}
\def\stepsl{\stepsto_{\lambda}}
\def\mstepsto{\stepsto^{\ast}}
\def\mstepsk{\mstepsto_{\kappa}}
\def\mstepsl{\mstepsto_{\lambda}}
\def\leadstok{\leadsto_{\kappa}}
\def\leadstol{\leadsto_{\lambda}}

\def\ruleone{
  \inference[%
      \textsf{[CStep]}
    ]{
      \CC \xleadsto{~\tau~} \mathrel{\vphantom{\to}_{\kappa}} \CC'
  }{
      \CC \rhd \state{\mu, p} \stepsto_{\kappa} \CC' \rhd \state{\mu, p}
    }
}

\def\ruletwo{
  \inference[
      \textsf{[Upd]}
    ]{
      r \in \replicas \quad \CC 
      \xleadsto{(\rid ,\upd[\op] , o)} {\! {}_{\kappa}} \CC'
  }{
      \CC \rhd \state{\mu ,\mathsf{upd}[\op]} \stepsto_{\kappa} \CC' \rhd \state{\mu ,\mathsf{skip}}
  }
}

\def\rulethree{
  \inference[%
    \textsf{[Qry]}
  ]{
      r \in \replicas & 
      \CC \xleadsto{(\rid , \qry[q] , \resp[v])} {\! {}_{\kappa}} \CC &
      \mu' = \mu[x \mapsto v]
  }{
      \CC \rhd \state{\mu , \mathsf{qry}[x](q)} \stepsto_{\kappa} \CC \rhd \state{\mu' ,\mathsf{skip}}
  }
}

\def\rulefour{
  \inference[%
    \mathsf{[\textsf{Comp}_1]}
    ]{
      \CC \rhd \state{\mu, p} \stepsto_{\kappa} \CC' \rhd  \state{\mu', p'}
  }{
      \CC  \rhd \state{\mu, p \fatsemi q} \stepsto_{\kappa} \CC'  \rhd \state{\mu' , p' \fatsemi q}
  }
}

\def\rulefive{
  \inference[%
    \mathsf{[\textsf{Comp}_2]}
    ]{
      \CC \rhd \state{\mu, p} \Downarrow_{\kappa}
  }{
      \CC \rhd \state{\mu, p \fatsemi q} \stepsto_{\kappa} \CC \rhd \state{\mu , q}
  }
}

\def\rulesix{
  \inference[%
    \textsf{[WDone]}
    ]{
      \eval{e}(\mu) = 0
  }{
    \CC \rhd \state{\mu , \mathsf{while}(e,p)} \Downarrow_{\kappa}
  }
}

\def\ruleseven{
  \inference[%
   \textsf{[WStep]}
  ]{
      \eval{e}(\mu) \neq 0
    }{
      \CC \rhd \state{\mu ,\mathsf{while}(e,p)} \stepsto_{\kappa} \CC \rhd \state{\mu , p \fatsemi \mathsf{while}(e,p)}
    }
}

\def\ruleeight{
  \inference[%
    \textsf{[Skip]}
    ]{
  }{
    \CC \rhd \state{\mu, \mathsf{skip}} \Downarrow_{\kappa}
  }
}

\def\rulenine{
  \inference[%
    \textsf{[Asn]}
    ]{
      \eval{e}(\mu) = v    &     \mu' = \mu[x \mapsto v]
  }{
      \CC \rhd \state{\mu , \mathsf{asn}[x](e)} \stepsto_{\kappa} \CC  \rhd \state{ \mu' , \mathsf{skip}}
  }
}

\section{A Representation-Independent Client Interface to CRDTs}
\label{sec:representation-independence}

In \cref{sec:emulation}, we view CRDTs as open systems,
in the sense that the client using a CRDT is an abstract entity
living outside the CRDT semantics, interfacing with them by
the labels, e.g., $\upd[\opsf]$, $\qry[q]$. In this section,
we `close' our system model by designing a small imperative
programming language of client programs, which maintains a store
that is modified by interactions with an underlying CRDT.
We use our results in \cref{sec:emulation}
to show how we can achieve a `contextual approximation' result: one can swap out an underlying host CRDT implementation
for its corresponding guest CRDT implementation, without seriously
affecting the results of the program.
We formalized our client language and proved our key theorem correct using Agda.\footnote{The code can be found at \url{https://zenodo.org/records/15866358}.}

\begin{assumptions}
  For the purposes of this section, we fix the following data:
  \begin{enumerate}
  \item A countably infinite set $\Var$ of program variables, and a set $\Expr$ of arithmetic expressions.
  \item The replica IDs $\replicas$, along with their operations $\Op$, and internal state space $S$.
  \item A query function $q : S \to \mathbb{N}$, and an evaluation function $\eval{-} : \Expr \to \mathbb{N}^{\Var} \to \mathbb{N}$
  \end{enumerate}
\end{assumptions}

The set of \emph{client programs} $\Prog$ is generated by the following grammar for
$x \in \Var$, $e \in \Expr$, $\op \in \Op$, $q \in Q$:
\begin{equation}
  \Prog \ni p,q \Coloneqq \mathsf{skip}
  \mid \mathsf{asn}[x](e) \mid \mathsf{while}(e,p) \mid p \fatsemi q
  \mid \mathsf{upd}[\op]
  \mid \mathsf{qry}[x](q)
\end{equation}

We think of $p \in \Prog$ as a client program running on top of a given CRDT system.
In $p$, the choice of replica that serves client requests
($\upd$ and $\qry$)
is made by external factors outside the client's control
(e.g., latency, availability, phases of the moon, etc.).
In that light, the client is served by a replica chosen in a non-deterministic
fashion.

\Cref{fig:eval-rules} gives the operational semantics of client programs.
To summarize, each rule in \cref{fig:eval-rules} is parameterized
by a configuration $\CC$ of the underlying CRDT system
(e.g., $\CC = \config{\Gamma}{\Sigma}{\beta}$, as in \cref{sec:coalg,sec:emulation}), and a 
variable store $\mu \in \mathbb{N}^{\Var}$. The configuration $\CC$
represents the \emph{execution environment} of the program state $\state{\mu, p}$.
We write $\CC \xleadsto{\alpha} {\! {}_{\kappa}} \CC'$ to denote a transition
of configurations by action $\alpha$, under CRDT system $\kappa$.
The transitions the CRDT system $\kappa$ may take are those defined in \cref{sec:coalg}.
On the other hand, client programs have two types of transitions,
namely
\[\CC \rhd \state{\mu, p} \stepsto_{\kappa} \CC' \rhd \state{\mu' , p'}
  \qquad \text{and} \qquad
  \CC \rhd \state{\mu, p} \Downarrow_{\kappa}.\]

The first arrow is read as \emph{on environment $\CC$ (and under
system $\kappa$), client program state $\state{\mu, p}$ progresses to 
$\state{\mu', p'}$ producing a new
environment $\CC'$}; the second arrow is read as \emph{on
environment $\CC$ (and under system $\kappa$), client program state $\state{\mu, p}$ terminates}.
On each computation step, the program
produces a new execution environment, reflecting changes in the internal
state of some replica $\rid$ or in the local store $\mu$ of the program state.
We also define the arrow $\Downarrow^{\ast}_{\kappa}$ so that
\[ \CC \rhd \state{\mu, p} \Downarrow^{\ast}_{\kappa} \iff
\exists \CC', \mu',  p'\, :\, 
( \CC \rhd \state{\mu, p} \mstepsk \CC' \rhd \state{\mu',p'}) 
\land (\CC' \rhd \state{\mu', p'} \Downarrow_{\kappa}).\]

\begin{figure}
  \scalebox{0.88}{$
    \begin{gathered}
      \ruleone \quad \ruletwo
      \\[1ex]
      \rulethree
      \\[1ex]
      \rulefour \quad \rulefive
      \\[1ex]
      \rulesix \quad \ruleseven
      \\[1ex]
      \ruleeight \quad \rulenine
    \end{gathered}
  $}
  \caption{Operational semantics of client programs.}
  \label{fig:eval-rules}
\end{figure}

We now define a notion of what it means for two (potentially different)
CRDT systems to ``approximate'' each other from the point of view of client programs.

\begin{definition}[CRDT Approximation]
\label{def:ctx-approx}
  Let $\kappa$ and $\lambda$ be two CRDT systems with resp.
  configurations $\CC_{\kappa}$ and $\CC_{\lambda}$. Let $\state{\mu_{1}, p_{1}}$
  and $\state{\mu_{2}, p_{2}}$
  be program states. We
  say $\CC_{\lambda} \rhd \state{\mu_1, p_1}$ \emph{approximates} 
  $\CC_{\kappa} \rhd \state{\mu_2, p_2}$ 
  (and we write $\CC_{\kappa} \rhd \state{\mu_1, p_1} \sqsubseteq \CC_{\lambda} \rhd \state{\mu_2, p_2}$ in that case)
  if we have
  \[
    \CC_{\kappa} \rhd \state{\mu_1, p_1} \Downarrow^{\ast}_{\kappa}
      \implies
    \CC_{\lambda} \rhd \state{\mu_2, p_2} \Downarrow^{\ast}_{\lambda}.
  \]
\end{definition}

\cref{def:ctx-approx} is essentially the coarsest form of comparison one would expect
from two program states executed in different CRDT environments. Indeed, it is coarser than at least weak simulation.

\begin{theorem}[Soundness]\label{thm:soundness}
  Let $\kappa$ and $\lambda$ be two arbitrary CRDT systems, and let
  $\CC_{\kappa}, \CC_{\lambda}$ be configurations of these respective systems. Then, for all program states $\state{\mu, p}$,
  \[ \CC_{\kappa} \wsimulatedby \CC_{\lambda} \implies 
  \CC_{\kappa} \rhd \state{\mu, p} \sqsubseteq \CC_{\lambda} \rhd \state{\mu, p}.\]
\end{theorem}

This key theorem tells us that if a configuration of a CRDT system $\kappa$ weakly simulates a configuration of a CRDT system $\lambda$, then a given client program running in the execution environment of the former is an approximation of that client program running in the execution environment of the latter.

We formalized the client program semantics shown in \cref{fig:eval-rules} and proved \cref{thm:soundness} correct in Agda.
Since our definition of emulation (\cref{def:emulation}) means that there exist a pair of emulations $\GG_1$, $\GG_2$
from resp. op-based to state-based and state-based to op-based CRDTs (\cref{sec:emulation}),
we know that CRDT emulation is sound with respect to \cref{def:ctx-approx}.
We summarize as a corollary to \cref{thm:soundness}.

\begin{corollary}[CRDT Emulation Is Sound]\label{corr:final}
  Let $\opsem{\OO}_1$ and $\stsem{\OO}_2$ resp. be an op-based CRDT and state-based CRDT with
  resp. transition semantics $\opleadsto_1$, $\stleadsto_2$ and
  resp. initial configurations $\opsem{\initstate}_1$ and $\stsem{\initstate}_2$.
  
  Then there are a pair of emulations $\GG_{1}$ and $\GG_2$,
  and a pair of resp. state-based CRDT $\stsem{\OO}'_{1}$ and op-based CRDT $\opsem{\OO}'_{2}$
  with resp. transition semantics $\stleadsto_{\GG_1}$, $\opleadsto_{\GG_2}$,
  and resp. initial configurations $\stsem{\initstate}_{\GG_1}$, $\opsem{\initstate}_{\GG_2}$,
  such that:
  \[
    \stsem{\OO}'_{1} = \GG_{1}(\opsem{\OO}_1) \textit{ and } \opsem{\initstate}_1 \wsimulatedby \wsimulates \stsem{\initstate}_{\GG_1}
    \quad \textit{and} \quad
    \opsem{\OO}'_{2} = \GG_2(\stsem{\OO}_2) \textit{ and } \stsem{\initstate}_2 \wsimulatedby \wsimulates \opsem{\initstate}_{\GG_2}.
  \]
  Therefore, for all program states $\state{\mu, p}$, we have
  \begin{enumerate}
    \item $\opsem{\initstate}_1 \rhd \state{\mu, p} \sqsubseteq \stsem{\initstate}_{\GG_1} \rhd \state{\mu, p}$
          and
          $ \stsem{\initstate}_{\GG_1} \rhd \state{\mu, p} \sqsubseteq  \opsem{\initstate}_1 \rhd \state{\mu, p}$;
    \item  $\stsem{\initstate}_2 \rhd \state{\mu, p} \sqsubseteq  \opsem{\initstate}_{\GG_2} \rhd \state{\mu, p}$ and
          $ \opsem{\initstate}_{\GG_2} \rhd \state{\mu, p} \sqsubseteq   \stsem{\initstate}_2 \rhd \state{\mu, p}$.
  \end{enumerate}
\end{corollary}

\cref{corr:final} concretely shows how our transition semantics of CRDTs in \cref{sec:coalg}
interfaces with the language we developed in this section. In words, it says that
for all client programs $p$ and all variable stores $\mu$, one can execute the
program state $\state{\mu, p}$ in an environment with either an op-based CRDT, or a state-based CRDT,
and so long as those CRDTs are related by an emulation $\GG$ (and we begin at the initial CRDT configurations),
termination behavior of $\state{\mu, p}$ in one execution environment implies termination behavior in the other,
and vice versa.

While cotermination is a fairly coarse form of equivalence, it is in some ways, stronger than one would expect.
For example, notice that in the \textsf{[Qry]} rule, interacting with the CRDT environment can change the variable store $\mu$.
Since programs $p$ are capable of infinite loops in our language, what \cref{corr:final} says is that,
if $p$ can terminate with $\OO$, then it still can with the $\GG(\OO)$ and vice-versa,
even if $p$ contains an arbitrarily large number of interactions with the CRDT via queries,
upon which termination might depend.
Of course, there is an element of non-determinism here: executing $\state{\mu, p}$ in one CRDT
environment and observing some behavior $b$ does not \emph{guarantee} that we will get $b$
in a different CRDT environment, only that $b$ is \emph{possible}, due to the inherent non-determinism of distributed systems.

\section{Related Work}
\label{sec:related}

As we discussed in \Cref{sec:introduction}, most existing research on CRDT verification has focused on verifying strong convergence 
and other safety\footnote{Notably, \citep{timany-trillium} also consider verification of \emph{liveness} properties, 
such as eventual delivery of messages.} 
properties for either  state-based~\cite{zeller-state-based-verification,gadduci-crdt-semantics,nair-state-based-verification,timany-trillium,nieto-aneris-state-based,laddad-crdt-synthesis} 
or op-based~\cite{gomes-verifying-sec,nagar-jagannathan-automated-crdt-verification,liu-lh-crdts,liang-feng-acc,nieto-aneris-op-based} CRDTs. 
One exception is the work of \citet{burckhardt-rdts-svo}, who give a framework for axiomatic specification and verification of 
both op-based and state-based CRDTs.

To our knowledge, our work is the first to make precise the sense in which state-based and op-based CRDTs emulate each other. 
However, the use of (bi)simulation relations in CRDT verification is not new.  For instance, \citet{burckhardt-rdts-svo}'s 
framework (discussed earlier in \Cref{subsec:refinement}) is based on \emph{replication-aware simulations}, and \citet{nair-state-based-verification} use a strong bisimulation 
argument to justify the use of simpler, easier-to-implement proof rules in an automated verification tool for state-based CRDTs, 
using the more complicated semantics as a reference implementation.

\citet{nieto-aneris-state-based} astutely observe that whether a CRDT is op-based or 
state-based is an implementation detail that should be hidden from clients.  They show that for a specific CRDT, the \emph{pn-counter},
a particular client program cannot distinguish between handwritten (rather than produced by emulation algorithms) op-based and state-based implementations of the CRDT. 
Although \citeauthor{nieto-aneris-state-based} do not consider the use of emulation algorithms, their work inspired our representation-independence result.  Rather than considering a specific pair of CRDT implementations 
and specific client program, though, we make precise the sense in which \emph{general} CRDT emulation 
algorithms result in the same observable behavior for the original and the emulating object.

Mechanized verification of CRDTs, both interactive~\cite{zeller-state-based-verification,gomes-verifying-sec,timany-trillium,nieto-aneris-op-based,nieto-aneris-state-based} 
and automated~\cite{nair-state-based-verification,nagar-jagannathan-automated-crdt-verification,deporre-verifx,laddad-crdt-synthesis}, is an active 
area of research with many exciting developments. In our work, we do not aim to provide user-ready verified implementations; 
our goal instead is to put existing efforts on a firm theoretical foundation, by making precise the sense in which results for 
op-based CRDTs can be said to transfer to state-based CRDTs and vice versa.


\section{Conclusion and Future Work}
\label{sec:conclusion}

In this paper, we have made precise the sense in which op-based and state-based CRDT systems emulate each others' behavior, using formal  simulation techniques. We have characterized which properties are preserved by CRDT emulation:
properties on weak traces. We have shown that CRDT emulation is, in some sense, sensitive to the underlying network semantics.
Nevertheless, we show that if two CRDT systems have starting configurations
that are related by a pair of simulations, then a client that interacts with either system cannot distinguish between the original system and the emulator, so long as the client interacts through an interface that does not expose details of the network semantics. We formalized and
proved this result correct in Agda, closing a long-standing gap in the CRDT literature.
Our results give researchers working with CRDTs a rigorous way to think about equivalence of state-based and op-based CRDTs,
both abstractly and theoretically in terms of simulation relations, and more concretely, in terms of interfaces to CRDTs.

\section*{Acknowledgements}\label{sec:acks}

We thank the anonymous reviewers of ICFP '25. Tobias Kapp\'e gave thoughtful feedback on an early draft. 
Stelios Tsampas acknowledges funding by the Deutsche Forschungsgemeinschaft (DFG, German Research Foundation), 
project number 527481841. 
This material is based upon work supported by the National Science Foundation under Grant No. 2145367. 
Any opinions, findings, and conclusions or recommendations expressed in this material 
are those of the author(s) and do not necessarily reflect the views of the National Science Foundation.

\bibliography{references}

\appendix
\crefalias{section}{appendix}
\section{Emulation From State-Based CRDTs to Op-Based CRDTs}\label{appendix:other-emulation}

In the main text of this paper, we formalized the emulation from op-based CRDTs to state-based CRDTs.
In this section, we show the other direction.

\begin{figure}[h]
  \centering
    \begin{align*}
        \QQ_{1} = \{ & (\stconfigA, \opconfigA) \mid \\
        & \textit{--- configurations need to be reachable } \\
        & \quad \stconfigA, \opconfigA \text{  are well-formed  } \\
        & \text{--- Agreement on global states} \\
        & \land \quad \stsem{\S} = \opsem{\S} \\
        & \textit{--- the op-based replica can simulate a ``large'' merge by the state-based replica}\\
        & \land \quad \forall (\rid, s') \in \stsem{\B}, \exists C \in \mathsf{Mergeable}_{\opsem{\B}}(\rid)\, :\,
            \stsem{\S(\rid)} \sqcup s' = \textstyle \bigsqcup (C \cup \{\opsem{\S(\rid)}\})
        \}
    \end{align*}
  \caption{Weak simulation relation of the host CRDT by the guest CRDT.}\label{fig:Q1}
\end{figure}

\begin{assumption}
  Assume we are given a state-based CRDT object as in \cref{fig:st-based-object},
  with its corresponding sets $S$, $\Op$, $Q$ and $V$, related functions,
  and semantics as in \cref{fig:st-global-rules}. That means
  $S = (S, \sqcup)$ is a join-semilattice with some initial state $\stsem{s^{0}}$, and $\upmap$ is inflationary, i.e.,
  $\forall \op \in \Op$, we have
    \[ s \sqcup \upmap(\rid, s, \op) = \upmap(\rid, s, \op).\]
  For brevity, we just write $\sqcup$ for $\mmap$.
\end{assumption}

The candidate emulation we consider is a mapping $\GG$ which takes our hosting
state-based CRDT object $\stsem{O}$ and constructs an op-based guest CRDT object
$\opsem{O'} = \GG(\stsem{O})$ as specified in \cref{fig:op-based-guest}.

\begin{figure*}[h!]
  \begin{align*}
        & \textbf{Parameters :} \\
        & \quad S : \text{states,}\; \Op : \text{operations,}\;
        S : \text{messages,}\; Q : \text{queries},\; V : \text{values,}\; \\
        & \quad  \stsem{s^{0}} : \text{ initial state } \\
        & \textbf{ Functions :}\\
        & \quad \opprmap(\rid, s, \op) = \stupmap(\rid, s, \op)\\
        & \quad \opemap(s, m) = s \sqcup m \qquad (\text{since $m$ is a message, $m \in S$})\\
        & \quad \opqmap(q, s) = \stqmap(q, s).
  \end{align*}
  \caption{The state-based guest CRDT constructed by $\GG$ from a given op-based host CRDT.}
  \label{fig:op-based-guest}
\end{figure*}

\paragraph{The Guest CRDT Simulates the Host CRDT}

Note that the roles of op-based and state-based CRDTs are reversed now,
so the candidate weak simulation $\QQ_{1}$ we are about to introduce
is actually analogous to $\RR_{2}$ in \cref{fig:R1-R2-sim}. Along those lines,
the main thing we deal with is asserting that a ``larger'' merge can be simulated
by a series of ``smaller'' merges.

\begin{definition}\label{def:mergeable-set}
  In a set $\beta \in \Pow(\replicas \times S)$, we say a set
    \[ C = \{s_1, \dots, s_k\} \subseteq S\]
  is a \emph{mergeable set at replica $\rid$} and write $C \in \mathsf{Mergeable}_{\beta}(\rid)$
  if for all $j \in 1..k$,
  we have $(\rid, s_j) \in \beta$.
\end{definition}

We take as our candidate simulation $\QQ_{1}$ the relation,
which clearly includes the initial configurations $\stconfig{\varepsilon}{\lambda \rid\, .\,  s^{0}}{\varnothing}$
and $\opconfig{\varepsilon}{\lambda \rid\, .\, s^{0}}{\varnothing}$.

\begin{theorem}\label{thm:main3}
  $\QQ_1$ is a weak simulation of $\stleadsto$ by $\opleadsto \subG$, which contains the initial configurations.
\end{theorem}

\paragraph{The Host CRDT Simulates the Guest CRDT}

Because the underlying replica states are both elements in $S$,
simulating broadcast with the state-based host CRDT means
that at all times, the global states and network states
are identical at every step of the simulation. This makes this direction
rather straightforward. 

We take as our candidate simulation $\QQ_{2}$ the relation
clearly includes the initial configurations $\opconfig{\varepsilon}{\lambda \rid\, .\, s^{0}}{\varnothing}$
and $\stconfig{\varepsilon}{\lambda \rid\, .\, s^{0}}{\varnothing}$.

\begin{figure}[h!]
  \centering
    \begin{align*}
      \QQ_{2} = \{ & (\opconfigA, \stconfigA) \mid \\
      & \textit{--- configurations need to be reachable} \\
      & \quad \opconfigA, \stconfigA \text{  are well-formed } \\
      & \textit{--- agreement in global states, and network states}\\
      & \land \quad \opsem{\S} = \stsem{\S} \quad \land \quad \opsem{\B} = \stsem{\B}
      \}.
    \end{align*}
  \caption{Weak simulation relation of the guest CRDT by the host CRDT.}\label{fig:Q2}
\end{figure}
And our statement of the theorem.

\begin{theorem}\label{thm:main4}
    $\QQ_2$ is a weak simulation of $\opleadsto \subG$ by $\stleadsto$, which contains the initial configurations.
\end{theorem}
\section{Proofs}\label{appendix:proofs}

\subsection{Proof of \cref{thm:main1}}

We show the omitted cases. Recall the hypothesis:
    \[ \RR_1(\opconfigA, \stconfigB) \textit{ and for any $\alpha$, we have} \opconfigA \xopleadsto{\alpha} \opconfigB.\]

\begin{itemize}
    \item (\OpUpdate). Then $\alpha = (\rid , \upd[\op] , \bot)$, and we have the transition
            where
              \begin{gather*}
                \opconfigA \xopleadsto{(\rid , \upd[\op] , \bot)} \opconfigB,\\
                \opsem{\Gamma'} = \opsem{\Gamma \cdot (\rid, \upd[\op] , \send[m])}, \quad
                \opsem{\Sigma'} = \opsem{\Sigma[\rid \mapsto s']}, \quad
                \opsem{\beta'} = \opsem{\bcast(\rid, m)(\beta)}.
              \end{gather*}

            We match with\StUpdate followed by\StSend as follows. First we construct the local transitions,
            \begin{gather*}
              \stsem{\S(\rid)} \xststepsto{(\rid, \upd[\op], \bot)}
              \stsem{H'}
              \xststepsto{(\rid, \bot, \send[H'])} \stsem{H'},
            \end{gather*}
            where $\stsem{H'} = \stsem{\upmap(\rid, \op, \S(\rid))}$.
            This allows us to apply\StUpdate followed by\StSend so that we have
              \[ \stconfigA \xstleadsto{(\rid , \upd[\op] , \bot)}\subG \stconfigB \quietleadsto \subG \stconfigC, \]
            where
              \begin{gather*}
                \stsem{\G''} = \stsem{G} \cdot \stsem{(\rid, \upd[\op], \bot)} \cdot \stsem{(\rid, \bot, \send[H'])},\\
                \stsem{\S''} = \stsem{\S [ \rid \mapsto H' ]},
                \qquad
                \stsem{\B''} = \stsem{\bcast(\rid, H')(\B)}.
              \end{gather*}

            We need to show that $(\opconfigB, \stconfigC) \in \RR_{1}$ by checking the conditions in \cref{fig:R1-sim}.
            It suffices to consider just the $\rid$ which performed the update. The only two conditions of interest are:
            \[
              \opsem{\Sigma'(\rid)} = \stsem{\Sigma''(\rid)} \qquad \textit{and} \qquad
              \forall \rid \in \replicas\, :\, (\opsem{\rid}, m) \in \opsem{\beta'}
              \implies (\stsem{\rid}, (m)^{\darrow \opsem{\Gamma'}}) \in \stsem{\beta''}.
            \]

            We know $\opsem{\Gamma'} = \opsem{\Gamma \cdot (\rid, \upd[\op] , \send[m])}$,
            so by\OpUpdate  we know $m = \opsem{\prmap(\rid,\op,\S(\rid))}$
            was the most recently sent message in the op-based system.
            Using the hypothesis, and the definitions of\StUpdate and $\stsem{\upmap}$,
            we deduce that the corresponding state $\stsem{H'} = \stsem{\S''(\rid)}$ 
            sent in the state-based system is
            \[
              \stsem{\S''(\rid)}
              = \stsem{\S(\rid)} \cup \{ \opprmap(\stsem{\rid}, \op, \interp(\stsem{\S(\rid)})) \}
              = \stsem{\S(\rid)} \cup \{ \opsem{\prmap(\rid, \op, \S(\rid))} \}
              = \stsem{\S(\rid)} \cup \{ m \}.
            \]

            Note that $m$ is (by definition), the most recently broadcast message by replica $\opsem{\rid}$,
            and that $\opsem{\upmap}$ is implemented by $\opsem{\rid}$ by delivering $m$ to itself (via $\opsem{\emap}$).
            Thus, it follows from the fact $\stsem{\S(\rid)} = \mathsf{Delivered}(\rid, \opsem{\B})$
            that we have
              \[ \stsem{\S''(\rid)} = \stsem{\S(\rid)} \cup \{ m \} = \mathsf{Delivered}(\rid, \opsem{\G'}) = (m)^{\darrow \opsem{\G'}}.\]
            Since\StSend broadcast $\stsem{H'} = \stsem{\S''(\rid)}$ to each replica, we obtain
            \[
              \forall \rid \in \replicas\, :\, (\opsem{\rid}, m) \in \opsem{\beta'}
              \implies (\stsem{\rid}, (m)^{\darrow \opsem{\Gamma'}}) \in \stsem{\beta''},
            \]
            as desired. To finish, we now just need to check that $\opsem{\S'(\rid)} = \stsem{\S''(\rid)}$ is true.
            Note that since $m$ was the most recently broadcast message by $\opsem{\rid}$, there \emph{does not} exist any 
            $m' \in \mathsf{Sent}(\opsem{\Gamma'}) = \mathsf{Sent}(\stsem{\Gamma''})$ such that $m \prec m'$. From the
            definition of $\interp$ in conjunction with the hypothesis, the following series of equalities are immediate.
            \[
              \opsem{\S'(\rid)}
              = \opemap(m, \opsem{\S(\rid)})
              = \opemap(\opsem{m}, \interp(\stsem{\S(\rid)}))
              = \interp(\stsem{\S(\rid)} \cup \{ m\})
              = \stsem{\S''(\rid)}.  
            \]
           \stopcase

      \item (\OpQuery). Immediate from the hypothesis. \stopcase

\end{itemize}

\subsection{Proof of \cref{thm:main2}}

We show the omitted cases.
  \[ \text{Hypothesis:}\; \RR_{2}(\stconfigA, \opconfigA)\;
      \text{and we have}\; \stconfigA \xstleadsto{\alpha} \subG \stconfigB,\; \text{some}\; \alpha. \]
      \begin{itemize}
      
          \item (\StUpdate). Then $\alpha = (\rid , \upd[\op] , \bot)$.
                Let $\stsem{H'} = \stsem{\S(\rid)} \cup \{ \opprmap(\rid, \op, \interp(\stsem{\S(\rid)})) \}$. Then we have
                  \begin{gather*}
                    \stconfigA \xstleadsto{(\rid , \upd[\op] , \bot)} \subG \stconfigB \\
                    \stsem{\G'} = \stsem{\G \cdot (\rid, \upd[\op], \bot)}, \quad
                    \stsem{\S'} = \stsem{\S[ \rid \mapsto H]}, \quad
                    \stsem{\B'} = \stsem{\B}.
                  \end{gather*}

                  We match at op-based replica $\opsem{\rid}$ as follows.
                  Let $m = \opsem{\prmap(\rid, \op, \S(\rid))}$ and
                  \newline
                  $\opsem{s'} = \opsem{\upmap(\rid, \op, \S(\rid))}$, then construct the local transition
                  \[
                    \opsem{\S(\rid)} \xopstepsto{(\rid, \upd[\op], \send[m])} \opsem{s'}.
                  \]
                  We then invoke the\OpUpdate rule, and obtain
                  \begin{gather*}
                    \opconfigA \xopleadsto{(\rid , \upd[\op] , \bot)} \opconfigB \\
                    \opsem{\G'} = \opsem{\G} \cdot \opsem{(\rid, \upd[\op], \send[m])}, \qquad
                    \opsem{\S'} = \opsem{\S[\rid \mapsto s'] }, \qquad
                    \opsem{\B'} = \opsem{\bcast(\rid, m)(\B)}.
                  \end{gather*}
                  To see this was a simulating transition, we show $\RR_2(\stconfigB, \opconfigB)$.

                  From the hypothesis $\opsem{\S(\rid)} = \interp(\stsem{\S(\rid)})$,
                  it is immediate that
                    \[ m = \opprmap(\rid, \op, \interp(\stsem{H}))
                      \quad \textit{and} \quad
                      m = \opprmap(\rid, \op, \opsem{\S(\rid)}),\]
                  and therefore $\mathsf{Delivered}(\rid, \opsem{\G'}) = \stsem{\S(\rid)} \cup \{ m\} = \stsem{\S'(\rid)}$.
                  Moreover, by construction, $m$ was the most recently sent message, hence the set of causal predecessors
                  is
                  \[ \stsem{\S(\rid)} = \mathsf{Delivered}(\rid, \opsem{\G}) = (m)^{\darrow \opsem{\G}} \setminus \{m\},\]
                  which means $\forall m' \in \stsem{\S(\rid)}$, we have $m' \prec m$, and therefore
                  from the definition of $\interp$, we obtain
                  \[
                    \opsem{\S'(\rid)}
                    = \opemap(m, \opsem{\S(\rid)})
                    = \opemap(m, \interp(\stsem{\S(\rid)}))
                    = \interp(\stsem{\S(\rid)} \cup \{m\}),
                  \]
                  and since
                  \[ 
                    \interp(\stsem{\S(\rid)} \cup \{m\})
                    = \interp(\stsem{\S(\rid)} \cup \{ \opprmap(\rid, \op, \interp(\stsem{\S(\rid)})) \})
                    = \interp(\stsem{\S'(\rid)}),
                  \]
                  we obtain $\opsem{\S'(\rid)} = \interp(\stsem{\S'(\rid)})$ as desired.
                  The other conditions are immediate, hence we conclude $\RR_2(\stconfigB, \opconfigB)$. \stopcase

          \item (\StSend) Then, $\alpha = \quiet{\tau}$, and there is a replica $\stsem{\rid}$ which broadcast its
                state $\stsem{\S(\rid)}$ to all other replicas, whence
          \begin{gather*}
            \stconfigA \quietleadsto \subG \stconfigB \\
            \stsem{\G'} = \stsem{\G} \cdot \stsem{(\rid, \bot, \send[\S(\rid)])}, \qquad
            \stsem{\S'} = \stsem{\S}, \qquad
            \stsem{\B'} = \stsem{\bcast(\rid, \S(\rid))(\B)}.
          \end{gather*}
          We simulate with the reflexive step, i.e.,
            \[ \opconfigA \quietleadstostar \opconfigA. \]
          We show that $\RR_{2}(\stconfigB, \opconfigA)$.

          Most conditions of $\RR_{2}$ are immediate, so it suffices to just check that
          the most recently sent $\stsem{\S(\rid)}$ corresponds to a deliverable set $U$ at any other replica $\opsem{\rid'}$.
          That is,
          \[
            (\stsem{\rid'}, \stsem{\S(\rid)}) \in \stsem{\B'} 
            \implies \exists \opsem{U} \in \mathsf{deliverable}_{\opsem{\B}}(\opsem{\rid'})
          \]
          i.e., the op-based $\opsem{\rid'}$ can simulate the state-based $\stsem{\rid'}$ merging $\stsem{\S(\rid)}$.
          Critically, the hypothesis implies $\stsem{\S(\rid)} = \mathsf{Delivered}(\opsem{\rid}, \opsem{\G})$,
          meaning that all the messages contained in $\stsem{\S(\rid)}$ were either delivered externally by $\opsem{\rid}$,
          or generated locally and self-delivered during a local update (c.f. \cref{fig:op-based-object}).

          From the well-formedness of $\opconfigA$, each $m \in \mathsf{Delivered}(\opsem{\rid}, \opsem{\G})$ was generated and then
          broadcast by \emph{some} $\opsem{\rid''} \in \replicas$. Thus, every other $\opsem{\rid'}$ has either delivered $m$, or $m \in \opsem{\B(\rid')}$.
          Fixing $\opsem{\rid'}$, it follows that we can partition
          $\stsem{\S(\rid)}$ into sets $D_{\opsem{\rid'}}$ and $U_{\opsem{\rid'}}$ of delivered and undelivered
          messages (at $\opsem{\rid'}$).

          \paragraph{Claim} We can take $U = U_{\opsem{\rid'}}$ as our deliverable set by choosing an ordering
          $\{m_1, \dots, m_k\} = U$ (where $k = |U|$), and delivering each $m_i$ following this ordering.

          \paragraph{Proof of Claim} To witness, suppose $U$ is non-empty (otherwise the claim is immediate).
          Since $\mathsf{Delivered}(\opsem{\rid}, \opsem{\G})$ is the set of delivered messages, the well-formedness condition (\cref{def:well-formed})
          combined with \cref{prop:causal-order-enforcement} implies there exists a total order $<_{\opsem{\rid}}$ of the 
          messages in 
          \newline $\mathsf{Delivered}(\opsem{\rid}, \opsem{\G})$ such that
          $<_{\opsem{\rid}}$ is the \emph{delivery order} of $\opsem{\rid}$, and $<_{\opsem{\rid}}$ is \emph{consistent} with causality in the sense
          that
            \[ \forall m', m'' \in \mathsf{Delivered}(\opsem{\rid}, \opsem{\G})\, :\, \neg (m' <_{\opsem{\rid}} m'') \implies \neg (m' \prec m'').\]
          In other words, there is an enumeration (dependent on $\opsem{\rid}$) $m_1, m_2, \dots, m_N$ of 
          $\mathsf{Delivered}(\opsem{\rid}, \opsem{\G})$ that is consistent with the causal relation $\prec$.
          Since $D_{\opsem{\rid'}} \cup U_{\opsem{\rid'}} = \stsem{\S(\rid)} = \mathsf{Delivered}(\opsem{\rid}, \opsem{\G})$,
          it must be the case that there is a $p \leq N$ so that
          \[ \{m_1, \dots, m_{p-1}\} = D_{\opsem{\rid'}} \qquad \textit{and} \qquad \{m_{p}, \dots, m_{N}\} = U_{\opsem{\rid'}}\]
          (up to re-ordering of the concurrent messages).
          Since we know that delivering $\{m_{p}, \dots, m_{N}\}$ in the order prescribed by $j = p, p+1, \dots, N$ does not violate causality,
          we set $U = U_{\opsem{\rid'}}$. It follows that $U \in \mathsf{deliverable}_{\opsem{\B}}(\opsem{\rid'})$, 
          which is what we wanted to show. \stopcase

          By construction, delivering $U$ at replica $\opsem{\rid'}$ simulates a merge of $\stsem{\S(\rid)}$ at $\stsem{\S(\rid')}$, since
          \[ \stsem{\S(\rid')} \cup \stsem{\S(\rid)} 
            = \mathsf{Delivered}(\opsem{\rid'}, \opsem{\B}) \cup \stsem{\S(\rid)}
            = \mathsf{Delivered}(\opsem{\rid'}, \opsem{\B}) \cup U,
          \]
          using the fact that $\stsem{\S(\rid)} = D \cup U$ and $D \subseteq \mathsf{Delivered}(\opsem{\rid'}, \opsem{\B})$.
          \newline We conclude $\RR_{2}(\stconfigB, \opconfigA)$ in this case. \stopcase

          \item (\StQuery). Immediate from the hypothesis.

                \stopcase
      \end{itemize}

\subsection{Proof of \cref{thm:main3}}

We need to show that the op-based guest CRDT can simulate the state-based host CRDT.

\[\text{Hypothesis:}\; \QQ_{1}(\stconfigA, \opconfigA)\;
  \text{and we have}\; \stconfigA \xstleadsto{\alpha} \stconfigB,\; \text{some}\; \alpha.
  \]
We proceed by structural induction on the transition $\stconfigA \xstleadsto{\alpha} \stconfigB$.

\begin{itemize}
  \item (\StDeliver). Then $\alpha = \quiet{\tau}$ and for some state $s$,
        we have $\stconfigA \quietleadsto \stconfigB$ where
          \[  \stsem{\G'} = \stsem{\G \cdot (\rid, \dlvr[s], \bot)}, \quad
              \stsem{\S'} = \stsem{\S[ \rid \mapsto \mmap(\S(\rid), s)]}, \quad
              \stsem{\B'} = \stsem{\B \rmv{ (\rid, s) }},
          \]
        which means that state-based replica $\stsem{\rid}$ merged a state $s$ sent by some other replica.
        If $s \leq \stsem{\S(\rid)}$, then $\stsem{\mmap(\S(\rid), s) = \S(\rid)}$, and there is nothing to show,
        so suppose $\neg (s \leq \stsem{\S(\rid)})$. Since the state-space is a join-semilattice, by the hypothesis, there is
        a set $C \in \mathsf{mergeable}_{\opsem{\beta}}(\opsem{r})$ so that we have:
        \[ \stsem{\S(\rid)} \sqcup s = \bigsqcup (C \cup \{\opsem{\S(\rid)} \}).\]

        Let $C = \{s_1, \dots, s_k\}$ for some $k$. Since $C$ is mergeable,
        it follows each $s_i$ satisfies $(\opsem{\rid}, s_i) \in \opsem{\beta}$, and therefore we can
        match with $k$ applications of (\OpDeliver) with respect to replica $\opsem{\rid}$ giving:
          \begin{gather}
            \begin{gathered}
              \opconfigA \quietleadstostar \subG \opconfigB,\\
              \opsem{\G'} = \opsem{\G} \cdot \opsem{\langle (\rid, \dlvr[s_1], \bot), \dots, (\rid, \dlvr[s_k], \bot) \rangle},\\
              \opsem{\S'} = \opsem{\S[\rid \mapsto s]}, \quad
              \opsem{\B'} = \opsem{\B} \opsem{\rmv{(\rid, s_1), \dots, (\rid, s_k)}}.
            \end{gathered}
          \end{gather}
        It now follows immediately that the conditions in \cref{fig:Q1} hold, thus
        \newline
        $\stconfigB \QQ_{1} \opconfigB$ holds.
  \item (\StUpdate). This case is immediate, since if $\stsem{\rid}$ performs an update, then $\opsem{\rid}$ can perform
        an identical update with the (\OpUpdate) rule, and all the conditions in \cref{fig:Q1} still hold on the updated configurations.
  \item (\StSend). Then $\alpha = \quiet{\tau}$, and some replica $\stsem{\rid}$ broadcast its state $\stsem{\S(\rid)}$ to
        all other replicas $\stsem{\rid'} \neq \stsem{\rid}$, yielding
        $\stconfigA \quietleadsto \stconfigB$ where
          \[  \stsem{\G'} = \stsem{\G \cdot (\rid, \send[\S(\rid)], \bot)}, \quad
              \stsem{\S'} = \stsem{\S} \quad
              \stsem{\B'} = \bcast(\stsem{\rid}, \stsem{\S(\rid)})(\stsem{\B}).
          \]
        In this case, the op-based system matches with the reflexive step
        \[ \opconfigA \quietleadstostar \subG \opconfigA, \]
        and we need to show that
        \[ \stconfigB \QQ_1 \opconfigA. \]

        For this, it suffices to show that for any replica $\stsem{\rid'}$, we have the implication
        \begin{gather}
        (\stsem{\rid'}, \stsem{\S(\rid)}) \in \stsem{\B}\\
        \implies \exists C \in \mathsf{mergeable}_{\opsem{\B}}(\opsem{\rid}) 
          : \stsem{\S(\rid')} \sqcup \stsem{\S(\rid)} = \bigsqcup (C \cup \{\opsem{\S(\rid')} \}),
        \end{gather}
        as the other conditions are immediate.

        If $\S(\rid) \leq \S(\rid')$, then $\S(\rid) \sqcup \S(\rid') = \S(\rid')$, and in this case,
        we can take $C = \varnothing$. So, we consider when $\neg(\S(\rid) \leq \S(\rid'))$.
        From the hypothesis, we have that $\stconfigB$ and $\opconfigA$ are well-formed,
        thus we have
        \begin{gather}
          s^{0} \xopstepsto{(\rid, i_1, o_1)} s^{1} \xopstepsto{(\rid, i_2, o_2)} \cdots \xopstepsto{(\rid, i_k, o_k)} s^{k} = \opsem{\S(\rid)},
        \end{gather}
        i.e., the current state of $\opsem{\rid}$ is a result of a sequence of $k \geq 0$ events $\opsem{(\rid, i_1, o_1)}, \dots, \opsem{(\rid, i_k, o_k)}$.
        Without loss of generality, we can assume each event advances $\opsem{\rid}$'s state in the partial order
        (if not, we can just remove the event, then re-index). 
        Thus, each $j$ event has the form of either
        \[ \opsem{(\rid, \upd[\op_{j}], \send[s^{j}])} \quad \textit{or} \quad \opsem{(\rid, \dlvr[m_{j}], \bot)},\]
        where each message $m_{j}$ is a state, i.e., $m_j = s'$ for $s'$ dependent on $j$.

        Note that by construction, op-based deliver events invoke $\opemap$, which is implemented in terms of $\sqcup$:
          \[\forall s, m \in S\, :\, \opemap(s, m) = s \sqcup m, \]
        and op-based updates are composition of $\opprmap$ and $\opemap$, thus we have
        \[ \forall s \in S\,  \op \in \mathtt{Op} :\, \opupmap (\rid, \op, s) = \opemap(\opprmap(\rid, \op, s), s) = s \sqcup \opprmap(\rid, \op, s).\]

        In other words, we can view each $s^{j}$ as a join between $s^{j-1}$ and some message $m_j \in S$, where either $m_j$ was generated locally
        by an update, or delivered over the network. By induction:
          \[ s^{k} = s^{k-1} \sqcup m_{k} = s^{k-2} \sqcup m_{k-1} \sqcup m_{k} = \cdots = s^{0} \sqcup (m_1 \sqcup \cdots \sqcup m_{k-1} \sqcup m_{k}).\]

        Recall that $\opsem{\S(\rid)} = s^{k}$. Noting that $s^{0}$ is the bottom of the lattice, we obtain the equalities:
        \[ \opsem{\S(\rid)} = s^{0} \sqcup (m_1 \sqcup \cdots m_{k-1} \sqcup m_{k}) = m_1 \sqcup \cdots \sqcup m_{k-1} \sqcup m_{k}.\]
      
        Critically, each message $m_j$ was generated locally \emph{and then broadcast}, or it was \emph{delivered} (hence broadcast by some other replica).
        In other words, by well-formedness, any other replica $\opsem{\rid'}$ has `access' to each $m_j$. That is, each replica $\opsem{\rid'}$
        in the current configuration $\opconfigA$
        has delivered some (possibly empty) subset of $\{m_1, \dots, m_{k-1}, m_{k}\}$. Call this set $D$. The `undelivered' messages are 
        $U = \{m_1, \dots, m_{k-1}, m_{k}\} \setminus D$. It is clear that $U$ is a mergeable set, so we can take $C = U$. It follows from the fact that
        $(S, \sqcup)$ is a join-semilattice, and the hypothesis (in particular, that $\stsem{\S} = \opsem{\S}$) that we can write:
        \[ \bigsqcup (C \cup \{\opsem{\S(\rid')} \}) = \bigsqcup (\{m_1, \dots, m_k\} \cup \{\opsem{\S(\rid')}\})
        = \opsem{\S(\rid')} \sqcup (m_1 \sqcup \cdots \sqcup m_k)
        = \stsem{\S(\rid')} \sqcup \stsem{\S(\rid)},
        \]
        which is what we wanted to show. It follows immediately that \[ \stconfigB \QQ_1 \opconfigA. \] \stopcase

  \item (\StQuery). Immediate from the hypothesis. \stopcase
\end{itemize}

Having shown all cases, the proof is finished.

\subsection{Proof of \cref{thm:main4}}

Since $\QQ_2$ (defined in \cref{fig:Q2}) maintains $\opsem{\S} = \stsem{\S}$ and $\opsem{\B} = \stsem{\B}$,
the proof here is markedly simpler.

First, it is clear that there is a simulation of local updates. To see this,
note the two equalities:
  \[
    \opupmap(\rid, \op, s) = \opemap(\opprmap(\rid, \op, s), s)
    = s \sqcup \opprmap(\rid, \op, s),
  \]
and
  \[
    \opprmap(\rid, \op, s) = \stupmap(\rid, \op, s),
  \]
we thus obtain:
\[ \opupmap(\rid, \op, s) = s \sqcup \opprmap(\rid, \op, s) = s \sqcup \stupmap(\rid, \op, s) = \stupmap(\rid, \op, s)\]
using the fact that $\stupmap$ is inflationary. 

The equality $\opsem{\B} = \stsem{\B}$ implies that deliver events are always able to be simulated,
so the only thing we really need to check is that
indeed $\opsem{\B} = \stsem{\B}$ is preserved by the candidate simulation --- but this has to be the case,
since removing a delivered message from both $\opsem{\B}$ and $\stsem{\B}$ preserves the equality,
and state-based replica $\stsem{\rid}$ can always simulate an op-based update by broadcasting directly
after a local update.

\subsection{Proof of \cref{thm:bisim-possible}}

\def\XX{\mathcal{X}}


We give the candidate weak bisimulation relation. Recall
$(-)^{\darrow \opsem{\G}}$ constructs the down set of a message.

\begin{figure}[h!]
  \centering
    \begin{align*}
      \bowtie\, = \{ & (\opconfig{\G}{\S}{\B} , \stconfig{\G}{\S}{\B}) \mid \\
                    & \textit{ --- configurations agree on messages, and replicas agree on delivered/merged messages} \\
                    & \quad (\mathsf{Sent}(\opsem{\G}) = \textstyle \bigcup \mathsf{Sent}(\stsem{\G}))
                      \land
                      \forall \rid \in \replicas\, :\, \mathsf{Delivered}(\rid, \opsem{\G}) = \stsem{\S}(\rid) \\
                    & \textit{ --- replicas agree on their local states} \\
                    & \land \quad \forall \rid \in \replicas\, :\, \opsem{\S}(\rid) = \interp(\stsem{\S}(\rid)) \\
                    & \textit{--- a message is pending in the op-system iff there is a pending downset in the state-system} \\
                    & \land \quad \forall (\rid, m) \in \opsem{\B} \iff (\rid, (m)^{\darrow \opsem{\G}}) \in \stsem{\B}
                  \}.
    \end{align*}
  \caption{Weak simulation relation of the guest CRDT by the host CRDT.}\label{fig:bowtie}
\end{figure}

Recall that for $\bowtie$ to be a weak \emph{bisimulation},
we need to show for all configurations $\opconfigA$, $\stconfigA$,
and labels $\alpha$, if $\opconfigA \bowtie \stconfigA$, we have both

\begin{enumerate}
  \item $\opconfigA \xopleadsto{\alpha} 
  \opconfigB \implies \exists \stconfigB \textit{ s.t. } \stconfigA \stsem{\xleadstostar{\alpha}} \stconfigB$
  \newline
  and $\opconfigB \bowtie \stconfigB$;

  \item $\stconfigA \xstleadsto{\alpha} \stconfigB \implies \exists \opconfigB \textit{ s.t. } \opconfigA \opsem{\xleadstostar{\alpha}} \opconfigB$
  \newline
  and $\opconfigB \bowtie \stconfigB$;
\end{enumerate}

For the proof, suppose $\opconfigA \bowtie \stconfigA$ holds. For the case of an op-based transition
$\opconfigA \xopleadsto{\alpha} \opconfigB$, we can proceed identically to the proof of \cref{thm:main1}
by structural induction on the transition.

For the case of a state-based transition $\stconfigA \xstleadsto{\alpha} \stconfigB$, The proof follows by induction on the
transition, following the same pattern as the proof of \cref{thm:main2}, though there is no longer a \StSend rule to consider, as
now broadcasts happen atomically within a local update. We show how to prove this subcase in particular.

\begin{itemize}
  \item (\StUpdBC). Then, $\alpha = \stsem{(\rid, \upd[\op], \send[H])}$, where $\stsem{H} = \stsem{\upmap(\rid, \op, \S(\rid))}$.
        We match with (\OpUpdate) by setting $\opsem{s'} = \opsem{\upmap(\rid, \op, \S(\rid))}$ and constructing the local transition
          \[ \opsem{\S(\rid)} \xopstepsto{(\rid, \op, \send[m])} \opsem{s'}, \]
        where $\opsem{m} = \opsem{\prmap(\rid, \op, \S(\rid))}$, and therefore obtain
        \[
          \opsem{\G'} = \opsem{\G} \cdot \opsem{(\rid, \upd[\op], \send[m])}, \qquad
          \opsem{\S'} = \opsem{\S[\rid \mapsto s'] }, \qquad
          \opsem{\B'} = \opsem{\bcast(\rid, m)(\B)}.
        \]
        Following the details of the \StUpdate case in the proof of \cref{thm:main2}, one can show
        $\opsem{\S'(\rid)} = \interp(\stsem{\S'(\rid)})$,
        from which it follows that $\mathsf{Sent}(\opsem{\G'}) = \textstyle \bigcup \mathsf{Sent}(\stsem{\G'})$,
        and $\mathsf{Delivered}(\rid', \opsem{\G'}) = \stsem{\S'}(\rid')$, for all $\rid'$.
        It remains to show that
        \[ \forall (\rid, m) \in \opsem{\B'} \iff (\rid, (m)^{\darrow \opsem{\G'}}) \in \stsem{\B'}. \]
        It is sufficient to focus on the newly generated message $m$.

        Note that by construction, $m$ causally proceeds after all 
        $m' \in \mathsf{Delivered}(\opsem{\G'(\rid)})$ (including the `self-delivered' messages from local updates).
        Further, it must be the case that $\stsem{H} = \stsem{\S'(\rid)} = (m)^{\darrow \opsem{\G'}}$.
        This follows from the fact that $m$ was the most recently generated message, and op-based replicas deliver to themselves
        the messages generated from local updates.
        Thus,
        \[ \stsem{\S'(\rid)} = \mathsf{Delivered}(\rid, \opsem{\G'}) = (m)^{\darrow \opsem{\G'}}. \]
        It now immediately follows that $\forall (\rid, m) \in \opsem{\B'} \iff (\rid, (m)^{\darrow \opsem{\G'}}) \in \stsem{\B'}$,
        as desired. We conclude $\opconfigB \bowtie \stconfigB$. \stopcase
\end{itemize}

It now follows from \cref{def:bisimulation} that $\bowtie$ is a weak bisimulation.

\end{document}